\documentclass[12pt,draftcls,onecolumn]{IEEEtran/IEEEtran}   % cf. http://nd.edu/~ieeetac/information.html
\usepackage{cite}
\ifCLASSINFOpdf
  \usepackage[pdftex]{graphicx}
  \graphicspath{{graphics/}}
  \DeclareGraphicsExtensions{.pdf,.jpeg,.png}
\else
  \usepackage[dvips]{graphicx}
  \graphicspath{{graphics/}}
  \DeclareGraphicsExtensions{.eps}
\fi
\usepackage[cmex10]{amsmath}
\usepackage{ifthen}
\newboolean{TwoColEq}
\setboolean{TwoColEq}{false} % boolvar=true or false
\usepackage{amssymb}  % assumes amsmath package installed
\usepackage{amsthm}
\usepackage{enumerate}

\usepackage{color}
\usepackage[textwidth=2cm,textsize=small]{todonotes}

\theoremstyle{plain}
\newtheorem{thm}{Theorem}

\newtheorem{lem}{Lemma}

\theoremstyle{definition}
\newtheorem{definition}{Definition}
\newtheorem{exmp}{Example}%[section]

\theoremstyle{remark}
\newtheorem{rem}{Remark}
\newtheorem{remark}{Remark}

\newcommand{\be}{\begin{enumerate}}
\newcommand{\ee}{\end{enumerate}}

\newcommand{\Tr}{\text{Tr}}

\newcommand{\D}{\mathsf D}
\newcommand{\R}{\mathsf R}
\newcommand{\Adj}{\text{Adj}}

\newcommand{\Prob}{\mathbb P}

\begin{document}
\title{Privacy-Preserving Filtering for Event Streams}
\author{Jerome~Le~Ny,~\IEEEmembership{Member,~IEEE}% <-this % stops a space
\thanks{This work was supported by NSERC under Grant RGPIN-435905-13.
J. Le Ny is with the department of Electrical Engineering, Polytechnique Montreal, 
and with GERAD, Montreal, QC H3T 1J4, Canada. {\tt\small jerome.le-ny@polymtl.ca}}% <-this % stops a space
\thanks{Preliminary versions of some of the results contained in this paper were presented 
at CDC 2013 and CDC 2014 \cite{LeNy_CDC13_eventStreamDP,LeNy_CDC14_MIMOeventFiltering}.}% <-this % stops a space
}
%\thanks{Manuscript received April 19, 2005; revised January 11, 2007.}}


%\markboth{Submitted to the IEEE Transactions on Automatic Control}%
\markboth{}%
{Le Ny: Privacy-Preserving Filtering for Event Streams}



\maketitle

\begin{abstract}
%\boldmath
Many large-scale information systems such as intelligent transportation systems, 
smart grids or smart buildings collect data about the activities of their users to optimize 
their operations. To encourage participation and adoption of these systems, it is becoming 
increasingly important that the design process take privacy issues into consideration.
In a typical scenario, signals originate from many sensors capturing events 
involving the users, and several statistics of interest need to be continuously 
published in real-time.  
This paper considers the problem of providing differential privacy guarantees 
for such multi-input multi-output systems processing event streams. 
%and operating continuously. 
We show how to construct and optimize various extensions 
of the zero-forcing equalization mechanism, which we previously proposed for 
single-input single-output systems. Some of these extensions can take a 
model of the input signals into account.
%
%We discuss the privacy-utility tradeoffs\todo{really?}~ of our privacy-preserving 
%mechanisms, and describe an application to privately monitoring and forecasting 
%occupancy in a building equipped with a dense network of motion detection sensors. 
%\todo{adjust based on the actual simulations}
We illustrate our privacy-preserving filter design methodology through the problem 
of privately monitoring and forecasting occupancy in a building equipped with multiple 
motion detection sensors.
%, and analyzing the activity of a shared processing server.
\end{abstract}


\section{Introduction}

Privacy issues associated with social networking applications or monitoring and decision systems collecting personal data to operate
are receiving an increasing amount of attention \cite{Weber10_law, PCAST14_privacy}. Indeed, privacy concerns 
are already resulting in delays or cancellations in the deployment of some smart power grids, location-based services, or 
civilian unmanned aerial systems for example \cite{EPIC03_privacyCenter}.
In order to encourage the adoption of these systems, which can provide important societal benefits, new tools are needed 
to provide clear privacy protection guarantees and allow users to balance utility with privacy rigorously \cite{Chan03_securityPrivacy_SN}.

Since offering privacy guarantees for a system generally involves sacrificing some level of 
performance, evaluating the resulting trade-offs %rigorously 
requires a quantitative definition of privacy. Various such definitions have been proposed, such as 
disclosure risk \cite{Duncan86_disclosure} in statistics, $k$-anonymity \cite{Sweeney02_kAnon}, 
information-theoretic privacy \cite{Sankar11_privacyInfoTheoretic}, or conditions based on 
observability \cite{Xue13_securityAVN, Manitara_ECC13_privacyConsensus}.
However, in the last few years the notion of differential privacy has emerged essentially as a standard 
specification \cite{Dwork06_DPcalibration, Dwork_ICAL06_DP}. Intuitively, a system processing privacy-sensitive 
inputs from individuals is differentially private if its published outputs are not too sensitive to the data provided 
by any single participant. This definition is naturally linked to the notion of system gain for 
dynamical systems, see \cite{LeNy_DP_CDC12, LeNyDP2012_journalVersion}.
One operational advantage of differential privacy compared to other definitions is that it provides strong guarantees 
without involving the difficult task of modeling all the available auxiliary information that could be linked to the 
published outputs, despite the fact that unanticipated privacy breaches are typically due to the presence of 
this side information \cite{Sweeney02_kAnon, Narayanan08_netflixBreach, Calandrino11_privacyAttackCollabFilt}.

Differential privacy is a strong notion of privacy, but might require large perturbations to the
published results of an analysis in order to hide individuals' data. This is especially true for applications 
where users continuously contribute data over time, and it is thus important to carefully design real-time mechanisms 
that can limit the impact on system performance of differential privacy requirements. Previous work on designing differentially
private mechanisms for the publication of time series include \cite{Rastogi10_DPtimeSeries, Li11_DPcompressive},
but these mechanisms are not causal and hence not suited for real-time applications. 
The ZFE mechanism of Section \ref{section: ZFM} could also be interpreted as a dynamic, causal version of the 
matrix mechanism introduced in \cite{LiMiklau12_adaptiveMech} for static databases.
The papers \cite{Dwork10_DPcounter, Chan11_DPcontinuous, Bolot11_DPdecayingSums}
describe real-time differentially private mechanisms to approximate a few specific filters processing a stream of $0-1$ variables,
representing the occurence of events attributed to individuals. For example, \cite{Dwork10_DPcounter, Chan11_DPcontinuous} consider 
a private accumulator providing at each time the total number of events that occurred in the past. This paper is inspired 
by this scenario, and builds on our previous work on this problem in \cite[Section IV]{LeNy_DP_CDC12} \cite[Section VI]{LeNyDP2012_journalVersion}.
Here we extend our analysis in particular to multi-input multi-output (MIMO) linear time-invariant systems, which considerably 
broadens the applicability of the model to more common situations where multiple sensors monitor an environment 
and we wish to concurrently publish several statistics of interest.
An application example is that of analyzing spatio-temporal records provided by networks of simple
counting sensors, e.g., motion detectors in buildings or inductive-loop detectors in traffic information systems 
\cite{LeNy14_traffic}. 
The literature on the differentially private processing of multi-dimensional time series is still very limited, 
but includes \cite{Cao13_slidingWindow}, which considers a single-input multiple-output filter where each output channel
corresponds to a moving average filter with a different size for the averaging window,
%(considered from a very different point of view) 
 as well as \cite{Fan13_mimoSeries}, which discusses an application to traffic monitoring.
%\todo{Maybe move to a new section on related work}

To summarize, the contributions and organization of this paper are as follows. In Section \ref{section: statement}, 
we present a new generic scenario where we need to approximate a general MIMO linear time-invariant system by
a mechanism offering differential privacy guarantees for the input signals. The formal definitions necessary to state
the problem are also provided in that section. In Section \ref{section: sensitivity} we perform some preliminary
system sensitivity calculations that are necessary in the rest of the paper. Section \ref{section: ZFM} presents a
general approximation scheme for MIMO systems that provides differential privacy guarantees for the input signals.
The design methodology and performance of the privacy-preserving filter are illustrated in Section \ref{section: building monitoring}
in the context of a building occupancy estimation problem. 
Note that Sections III-V provide a more detailed presentation of the theoretical and simulation results contained in 
our conference paper \cite{LeNy_CDC14_MIMOeventFiltering}.
Finally, Section \ref{section: additional info} presents
additional privacy-preserving mechanisms that can approximate the desired outputs more closely but
require more information about the input signals to be publicly available, .e.g., their second-order statistics. 
It improves and extends to the MIMO case the results presented in our conference paper \cite{LeNy_CDC13_eventStreamDP}. 
This section also illustrates the relationship between our problem and certain joint Transmitter-Receiver optimization problems 
arising in the communication systems literature \cite{Salz:85:MIMOopt, Yang:94:TxRxOpt}.

%\textcolor{red}{[In the last section, we obtain optimization problems similar to the joing Tx-Rx optimization problems
%arising in the communication systems literature \cite{Salz:85:MIMOopt, Yang:94:TxRxOpt}, and for which 
%waterfilling-type solutions are typical \cite{Palomar:TSP05:waterfilling}.]}

%In addition, we describe mechanisms that can take into account certain known statistics about the input streams, e.g.,
%second order statistics for wide-sense stationary signals, and improve on the performance of the 
%Minimum Mean Square Error (MMSE) mechanism discussed in \cite[Section VI.A]{LeNyDP2012_journalVersion}.

\textbf{Notation:} 
Throughout the paper we use the following standard abbreviations: LTI for Linear Time-Invariant, SISO for Single-Input Single-Output, 
SIMO for Single-Input Multiple-Output, and MIMO for Multiple-Input Multiple-Output. Unless specified otherwise, dynamical systems
or filters are assumed causal, and transfer functions have real-valued coefficients.
We fix a base probability space $(\Omega, \mathcal F, \Prob)$. For $m$ an integer with $m \geq 1$, we write $[m] := \{1,\ldots,m\}$.
The notations $|x| _1= \sum_{k=1}^p |x_k|$ and $|x| _2= \left( \sum_{k=1}^p |x_k|^2 \right)^{1/2}$ are used to denote the 
%Euclidean norm for $x$ 
$1$- and $2$-norms in $\mathbb R^p$ or $\mathbb C^p$, and we reserve the notation $\| \cdot \|$ for norms on
signal and system spaces. $\text{col}(x_{1}, \ldots x_{p})$ denotes a column vector or signal
with components $x_i$, $i=1,\ldots,p$, and $\text{diag}(x_1,\ldots,x_m)$ denotes a diagonal $m \times m$ matrix
with the $x_i$'s on the diagonal. % and zeros elsewhere.}
%\textcolor{red}{[col is use in the proof of thm 3. If used nowhere else, remove from here.]}
Finally, for $H$ a Hermitian matrix, $H \succ 0$ means that it is positive definite, and $H \succeq 0$ that it is positive semi-definite.

%%%%%%%%%
% New paper
%%%%%%%%%

%SIMO case with sliding windows: \cite{Cao13_slidingWindow}.
%Not clear how to extend previous analyses for counts to general signal processing systems.

%Examples of applications: monitoring diseases and other health-care related issues, traffic information and 
%location-based services more generally, economics. In these scenarios, we have a record of the number of
%people with a certain characteristic in each region at each time. Each person can be in at most one region
%at each time.

%
%For spatial monitoring, we will assume that each person contributes to the dataset of each region only once. For example,
%it is counted when it enters the region (as with traffic detectors for example).
%
%Also, multiple inputs could come from breaking a single stream to which a user can contribute multiple times into 
%multiple streams, if we have say a bound on the inter-event time for the users. Example, motions through a building.
%
%The real-time requirement is thus motivated by the fact that we would like to integrate our filter outputs into 
%real-time information and decision-making systems, e.g., traffic information systems.

% !TEX root =  ../dpEventFiltering.tex

\section{Problem Statement}			\label{section: statement}

%\subsection{Application Examples}

%To do. Present one or two applications for which we have data, that will be used for simulations at the end.

\subsection{Generic Scenario}	\label{section: generic scenario}

We consider $m$ sensors detecting events, with sensor $i$ producing a discrete-time scalar signal $\{u_{i,t}\}_{t \geq 0} \in \mathbb R$, for $i \in [m]$. 
In a building monitoring scenario for example, %such as the one described in Section \ref{section: building monitoring}, 
the sensors could be motion detectors distributed at various locations and polled at regular intervals, with $u_{i,t} \in \mathbb N$ 
the number of detected events reported for period $t$.
We denote $u$ the resulting vector-valued signal, i.e., $u_t \in \mathbb R^m$.
A linear time-invariant (LTI) filter $F$, with $m$ inputs and $p$ outputs, takes input signals $u$ 
from the sensors and publishes output signals $y = Fu$ of interest, with $y_t \in \mathbb R^p$. 
%In other words, in the general case we consider systems with $m$ inputs and $p$ outputs. 
In our example, we might be interested in continuously updating real-time estimates of the number of people in
various parts of the building, as well as short- and medium-term occupancy forecasts, in order to optimize the 
operations of the Heating, Ventilation, and Air Conditioning (HVAC) system.
%
%Moreover, it is generally necessary to produce multiple analysis, i.e., multiple output signals. It is also required here
%that the outputs be produced in real-time, or perhaps with a small bounded delay. 
%
The problem considered in this paper consists in replacing the filter $F$ by a system processing the input $u$ and producing
a signal $\hat y$ as close as possible to the desired output $y$ (minimizing for example the mean squared error 
$
\lim_{T \to \infty} \frac{1}{T} \sum_{t=0}^\infty \mathbb E[|y_t - \hat y_t|_2^2]
$), 
while providing some privacy guarantees to the individuals from which the input signals $u$ originate. 
The privacy constraint is explained and quantified in the next subsection.

\subsection{Differential Privacy}	\label{section: DP def}

As mentioned in the introduction, a differentially private mechanism publishes data in a way that is not too sensitive to the presence 
or absence of a single individual. A formal definition of differential privacy is provided in Definition \ref{def: differential privacy original} below. 
In the previous building monitoring example, one goal of a privacy constraint could be to 
provide guarantees that an individual cannot be tracked too precisely from the published (typically aggregate) data. 
Indeed, Wilson and Atkeson \cite{Wilson05_tracking} for example demonstrate how to track individual users in a building using a network 
of simple binary sensors such as motion detectors.

\subsubsection{Adjacency Relation}		

%We could say that two multi-dimensional inputs $u, u'$ are adjacent if and only if there is a single time $k$ at which all coordinates 
%$u_i[k]$ differ by at most $1$. Think for example at input streams that count, for each period, the number of people with various characteristics
%(one coordinate for each characteristic). Then one person can change all input streams by $1$ or $0$, but at a single time (assuming each
%person can contribute its data only once).
%
%Or think of streams as distinct places being visited. Then a given person can contribute to all streams, but only at most one and
%all the contributions are at different times.

Formally, we start by defining a symmetric binary relation, denoted $\Adj$, on the space $\D$ of datasets of interest, which 
captures what it means for two datasets to differ by the data of a single individual.
Essentially, it is hard to determine from a differentially private output which of any two adjacent input datasets was used.
Here, $\D := \{u: \mathbb N \mapsto \mathbb R^m\}$ is the set of input signals, and we have $\Adj(u,u')$ if and only if 
we can obtain the signal $u'$ from $u$ by adding or subtracting the events corresponding to just one user. Motivated again 
by applications to spatial monitoring,
%motion detection, 
we consider in this paper the following adjacency relation 
\ifthenelse {\boolean{TwoColEq}} 
{
\begin{align}	\label{eq: adjacency relation}
&\text{Adj}(u,u') \text{ iff } \nonumber \\
&\forall i \in [m], \exists t_i \in \mathbb N, \alpha_i \in \mathbb R, \text{ s.t. } u'_i - u_i = \alpha_i \delta_{t_i}, |\alpha_i| \leq k_i,
\end{align}
}
{
\begin{align}	\label{eq: adjacency relation}
&\text{Adj}(u,u') \text{ iff } \forall i \in [m], \exists t_i \in \mathbb N, \alpha_i \in \mathbb R, \text{ s.t. } u'_i - u_i = \alpha_i \delta_{t_i}, |\alpha_i| \leq k_i,
\end{align}
}
parametrized by a vector $k \in \mathbb R^m$ with components $k_i > 0$. 
%\todo{note that $k_i = 0$ is an issue in the ZFE thms}
According to (\ref{eq: adjacency relation}), a single individual can affect each input signal component at a \emph{single time} (here $\delta_{t_i}$ denotes the
discrete impulse signal with impulse at $t_i$), and by at most $k_i$. % \in \mathbb R_+$. 
%We adopt the following notation in the rest of the paper. We denote by $k \in \mathbb R^m$ the vector with components $k_i$.
Let $e_i \in \mathbb R^m$ be the $i^{th}$ basis vector, i.e., with coordinates $e_{ij} = \delta_{ij}, j =1, \ldots, m$. Then for two adjacent signals $u$, $u'$, 
we have with the notation in (\ref{eq: adjacency relation})
\begin{equation}	\label{eq: adjacent signals}
u'-u = \sum_{i=1}^m \alpha_i \delta_{t_i} e_i.
\end{equation}
Note in passing that we can place 
additional constraints on $k$ to capture additional knowledge about the problem, which can help design mechanisms with better performance,
as we discuss later. 
For example, if we known that a given person can activate at most $l < m$ sensor and each $k_i$ is $1$, we can add the 
constraint $|k|_1 \leq l$. %$\|k\|_2 \leq \sqrt{l}$.}

The adjacency relation (\ref{eq: adjacency relation}) extends the one considered 
in \cite{Dwork10_DPcounter, Chan11_DPcontinuous, LeNyDP2012_journalVersion} to the case of multiple input signals.
It puts two constraints on the influence that an individual can have on
the input data in order for a differentially private mechanism to offer him guarantees. First, any given sensor can report an event
due to the presence of the individual only once over the time interval of interest for our analysis. This is a sensible constraint in 
applications such as traffic monitoring with fixed motion detectors activated only once by each car traveling along a road, certain 
location-based services where a customer would check-in say at most once per day at each visited store, or certain health-monitoring
applications where an individual would report a sickness only once.
For a building monitoring scenario however, a single user could trigger 
the same motion detector several times over a relatively short period. A first solution consists in splitting the data 
stream of problematic sensors into several successive intervals, each considered as the signal from a new virtual sensor, so that an 
individual's data is present only once in each interval. A MIMO mechanism can then process such data and offer guarantees, addressing one
of the main issues for the applicability of the model proposed in \cite{Dwork10_DPcounter, Chan11_DPcontinuous}. However, increasing
the number of inputs degrades the privacy guarantees or the output quality that we can provide. Hence in general no privacy 
guarantee will be offered to users who activate the same sensor too frequently.
The second constraint imposed by (\ref{eq: adjacency relation}) is that we bound the magnitude of an individual's contribution by $k_i$, 
but this is not really problematic in applications such as motion detection, where we can typically take $k_i = 1$. 

\subsubsection{Definition of Differential Privacy}

Mechanisms that are differentially private necessarily randomize their outputs, in such a way that they satisfy the following property
\cite{Dwork06_DPcalibration, Dwork_ICAL06_DP, Dwork06_DPgaussian}.

\begin{definition}	\label{def: differential privacy original}
Let $\D$ be a space equipped with a symmetric binary relation denoted $\Adj$, and let $(\R, \mathcal M)$ 
be a measurable space. Let $\epsilon, \delta \geq 0$. A mechanism $M: \D \times \Omega \to \R$ is 
$(\epsilon, \delta)$-differentially private for $\Adj$ (and $\mathcal M$) if for all $d,d' \in \D$ such that $\Adj(d,d')$, we have
\begin{align}	\label{eq: standard def approximate DP original}
\Prob(M(d) \in S) \leq e^{\epsilon} \Prob(M(d') \in S) + \delta, \;\; \forall S \in \mathcal M. 
\end{align}
If $\delta=0$, the mechanism is said to be $\epsilon$-differentially private. 
\end{definition}

This definition quantifies the allowed deviation for the output distribution of a differentially private mechanism for two
adjacent datasets $d$ and $d'$. One can also show that it is impossible to design a statistical test with small 
error to decide if $d$ or $d'$ was used by a differentially private mechanism to produce its output
\cite{Wasserman10_DPstatistics, Kairouz:13:DPcomposition}.
In this paper, the space $\D$ was defined as the space of input signals, and the adjacency relation considered is (\ref{eq: adjacency relation}). 
The output space $\R$ is simply the space of output signals $\R:=\{y: \mathbb N \to \mathbb R^p\}$. Finally, a differentially private mechanism 
will consist of a system approximating our MIMO filter of interest $F$, as well as a source of noise necessary to randomize the outputs and satisfy 
(\ref{def: differential privacy original}). %$(\ref{eq: standard def approximate DP original})$. 
We also refer the reader to \cite{LeNy_DP_CDC12} for a technical discussion on the (standard) $\sigma$-algebra $\mathcal M$ used on 
the output signal space to offer useful guarantees.

\subsubsection{Sensitivity}

Enforcing differential privacy can be done by randomly perturbing the published output of a system, at the price of reducing its utility or quality.
Hence, we are interested in evaluating as precisely as possible the amount of noise necessary to make a mechanism differentially private. For this purpose, 
the following quantity plays an important role. 
%Note that the notation $|x| = \left( \sum_{k=1}^p |x_k|^2 \right)^{1/2}$ is used throughout the paper to denote the Euclidean norm for $x$ in $\mathbb R^p$ or $\mathbb C^p$

\begin{definition}	\label{def: sensitivity}
The $\ell_2$-sensitivity of a system $G$ with $m$ inputs and $p$ outputs with respect to the adjacency relation $\Adj$ is defined by
\begin{align*}
\Delta^{m,p}_2 G &= \sup_{\text{Adj(u,u')}} \| Gu - Gu' \|_2 = \sup_{\text{Adj(u,u')}} \| G(u - u') \|_2, 
\end{align*}
%The second equality is Parceval identity.
where by definition $\| G v \|_2 = \sqrt{\sum_{t=-\infty}^\infty | (Gv)_t |_2^2}$. 
%\[
%\| G v \|_2^2 = \sum_{t=0}^\infty | [Gv]_t |^2 = \frac{1}{2 \pi} \int_0^{2 \pi} | G(e^{j \omega}) v(e^{j \omega}) |^2 d \omega,
%\]
%et $|x| = \left( \sum_{k=1}^p |x_k|^2 \right)$ denotes the Euclidean norm for $x \in \mathbb C^p$. 
\end{definition} 

\subsubsection{A Basic Differentially Private Mechanism}

The basic mechanism of Theorem \ref{eq: basic DP mechanism} below (see \cite{LeNyDP2012_journalVersion}), 
extending  \cite{Dwork06_DPgaussian}, can be used to answer queries in a differentially private way.
%A differentially private mechanism proposed in  modifies an answer to a 
%numerical query by adding iid zero-mean Gaussian noise.
%
To present the result, we recall first the definition of the $\mathcal Q$-function 
$\mathcal Q(x) := \frac{1}{\sqrt{2 \pi}} \int_x^{\infty} e^{-\frac{u^2}{2}} du$. Now for $\epsilon, \delta > 0$, %$\frac{1}{2} \geq \delta > 0$,
let $K = \mathcal Q^{-1}(\delta)$ and define $\kappa_{\delta,\epsilon} = \frac{1}{2 \epsilon} (K+\sqrt{K^2+2\epsilon})$.
%\textcolor{red}{[Check in the proof where we should have the condition $\delta \leq 1/2$.]}

\begin{thm}	\label{eq: basic DP mechanism}
Let $G$ be a system with $m$ inputs and $p$ outputs, and with $\ell_2$-sensitivity $\Delta^{m,p}_2 G$ with
respect to an adjacency relation $\Adj$. 
Then the mechanism $M(u) = Gu + w$, where $w$ is a $p$-dimensional Gaussian white noise with covariance matrix
$\kappa_{\delta,\epsilon}^2 (\Delta^{m,p}_2 G)^2 I_p$ is $(\epsilon,\delta)$-differentially private with respect to $\Adj$.
\end{thm}

%\vspace{0.1cm}
%\begin{thm}	\label{thm: Gaussian mech}
%Let $q: \D \to \mathbb R^k$ be a query.
%Then the Gaussian mechanism $M_q: \D \times \Omega \to \mathbb R^k$ 
%defined by $M_q(d) = q(d) + w$, with $w \sim \mathcal N\left(0,\sigma^2 I_k \right)$, 
%where $\sigma \geq \frac{\Delta_2 q}{2 \epsilon}(K + \sqrt{K^2+2\epsilon})$ and $K = \mathcal Q^{-1}(\delta)$,
%is $(\epsilon,\delta)$-differentially private.
%\end{thm}

The mechanism  $M$ described in Theorem \ref{eq: basic DP mechanism}, which is a differentially-private
version of a system $G$, is called an output-perturbation mechanism. We see that the amount of noise sufficient 
for differential privacy with this mechanism is proportional to the $\ell_2$-sensitivity 
of the filter and to $\kappa_{\delta,\epsilon}$, which can be shown to behave roughly as $O(\ln(1/\delta))^{1/2}/\epsilon$.
Note that we add noise proportional to the sensitivity of the whole filter $G$ independently on \emph{each} output, 
even if $G$ was diagonal say, otherwise trivial attacks that simply average a sufficient number of outputs could 
potentially detect the presence of an individual with high probability \cite{LeNyDP2012_journalVersion}.

In conclusion we could obtain a differentially private mechanism for our original problem by simply adding 
a sufficient amount of noise to the output of our desired filter $F$, provided we can compute its sensitivity, 
which is the topic of the next section. Moreover, it is possible in general to design mechanisms with much 
less overall noise than this output-perturbation scheme, as discussed in Sections \ref{section: ZFM} 
and \ref{section: additional info}.	

%\textcolor{red}{
%----
%To potentially include above, otherwise somewhere below on the performance index.
%----
%Throughout this paper, we measure 
%the precision of specific approximations, hence the utility of the released privacy-preserving output,  by the mean squared error (MSE) between 
%the published and desired outputs, i.e., 
%$
%\lim_{T \to \infty} \frac{1}{T} \sum_{t=0}^\infty \mathbb E[|e_t|^2],
%$
%with $e = y - \hat y$. The next section is devoted to the description of two ways of choosing the filters $G, H$ as linear filters.
%}

 % new paper

%!TEX root =  ../dpEventFiltering.tex

\section{Sensitivity Calculations}		\label{section: sensitivity}		%  and \\ Input Perturbation Mechanism}

For the following sensitivity calculations (see Definition \ref{def: sensitivity}), the $\mathcal H_2$ norm of an LTI system plays an important role. 
We recall its definition for a system $G$ with $m$ inputs
\[
\|G\|_2^2 = \sum_{i=1}^m \| G \delta_0 e_i\|_2^2 = \frac{1}{2 \pi} \int_{-\pi}^{\pi} \Tr(G^*(e^{j \omega}) G(e^{j \omega})) d \omega.
\]
Writing $G(z) = [G_{ij}(z)]_{i,j}$ for the $p \times m$ transfer matrix, we also note from the frequency domain definition that 
$
\|G\|_2^2 = \sum_{i,j} \|G_{ij}\|_2^2.
$

\subsection{Exact solutions for the SIMO and Diagonal Cases}

Generalizing the SISO scenario considered in \cite{LeNy_DP_CDC12, LeNyDP2012_journalVersion}
to the case of a SIMO system, 
%system with $m=1$ input but possibly multiple outputs (SIMO), 
we have immediately the following theorem.
\begin{thm}[SIMO LTI system]		\label{thm: SIMO sensitivity}
Let $G$ be a stable LTI system with one input and $p$ outputs. 
For the adjacency relation (\ref{eq: adjacency relation}), we have
$
\Delta^{1,p}_2 G = k_1 \|G\|_2,
$
where $\|G\|_2$ is the $\mathcal H_2$ norm of $G$.
\end{thm}

\begin{proof}
We have immediately for $u$ and $u'$ adjacent
\ifthenelse {\boolean{TwoColEq}} {
\begin{align*}
\|G (u - u') \|^2_2 &= |\alpha_1|^2 \| G \delta_{t_1} \|^2_2 \\
&\leq k^2_1 \| G \|^2_2,
\end{align*}
}
{
\begin{align*}
\|G (u - u') \|^2_2 &= |\alpha_1|^2 \| G \delta_{t_1} \|^2_2 \leq k^2_1 \| G \|^2_2,
\end{align*}
}
and the bound is attained if $|\alpha_1| = k_1$.
%where $g$ denotes the impulse response of $G$ (here $g_t \in \mathbb R^p, \forall t$).
\end{proof}

For a system $G$ with multiple inputs, the special case where $G$ is diagonal, i.e., its 
transfer matrix is $G(z) = \text{diag}(G_{11}(z), \ldots, G_{mm}(z))$, also leads to a simple sensitivity result.
Note that in this case, we have $\|G\|_2^2 = \sum_{i=1}^m \|G_{ii}\|^2$.
% immediate for example from the formula with trace, in the frequency domain

\begin{thm}[Diagonal LTI system] 	\label{thm: diagonal system sensitivity}
Let $G$ be a stable \emph{diagonal} LTI system with $m$ inputs and outputs. 
For the adjacency relation (\ref{eq: adjacency relation}), denoting $K = \text{diag}(k_1, \ldots, k_m)$,
we have
\[
\Delta^{m,m}_2 G = \|GK\|_2 = \left( \sum_{i=1}^m \|k_i G_{ii}\|_2^2 \right)^{1/2}.
\]
%where $K = \text{diag}(k_1, \ldots, k_m)$.
\end{thm}
\begin{proof}
If $G$ is diagonal, then for u and u' adjacent, we have from (\ref{eq: adjacent signals})
\ifthenelse {\boolean{TwoColEq}} {
\begin{align*}
\|G(u-u')\|^2_2 &= \left \| \sum_{i=1}^m \alpha_i G \delta_{t_i} e_i \right\|^2_2 \\
&= \| \text{col}(\alpha_1 g_{11} * \delta_{t_1}, \ldots, \alpha_m g_{mm} * \delta_{t_m}) \|^2_2,
\end{align*}
}
{
\begin{align*}
\|G(u-u')\|^2_2 &= \left \| \sum_{i=1}^m \alpha_i G \delta_{t_i} e_i \right\|^2_2 
= \| \text{col}(\alpha_1 g_{11} * \delta_{t_1}, \ldots, \alpha_m g_{mm} * \delta_{t_m}) \|^2_2,
\end{align*}
}
where %$\text{col}(x_{1}, \ldots x_{m})$ denotes a signal with values in $\mathbb R^m$ if each $x_i$ is a scalar signal, and 
$g_{ii}$ denotes the impulse response of $G_{ii}$. Hence
\ifthenelse {\boolean{TwoColEq}} {
\begin{align*}
\|G(u-u')\|^2_2 &= \sum_{i=1}^m \|\alpha_i g_{ii} * \delta_{t_i}\|_2^2  \\
&= \sum_{i=1}^m |\alpha_i|^2 \|G_{ii}\|_2^2, 
\end{align*}
}
{
\begin{align*}
\|G(u-u')\|^2_2 &= \sum_{i=1}^m \|\alpha_i g_{ii} * \delta_{t_i}\|_2^2  
= \sum_{i=1}^m |\alpha_i|^2 \|G_{ii}\|_2^2, 
\end{align*}
}
and $|\alpha_i| \leq k_i$, for all $i$. Again the bound is attained if $|\alpha_i| = k_i$ for all $i$.
\end{proof}

\subsection{Upper and Lower Bound for the general MIMO Case}

For MISO or general MIMO systems, the sensitivity calculations are no longer so straightforward, because the
impulses on the various input channels, obtained from the difference of two adjacent signals $u, u'$,
all possibly influence any given output. Still, the following result provides simple bounds on the sensitivity.%,
%which will be used later for mechanism optimization.

\begin{thm} 	\label{thm: bound MIMO general}
Let $G$ be a stable LTI system with $m$ inputs and $p$ outputs. 
%$p \times m$ transfer matrix $G(z)=[G_1(z),\ldots,G_m(z)]$ (i.e., with columns $G_i$).
%, such that $\|G\|_2 < \infty$. 
For the adjacency relation (\ref{eq: adjacency relation}), denoting $K = \text{diag}(k_1, \ldots, k_m)$,
and $|k|_2 = \left(\sum_{i=1}^m k_i^2\right)^{1/2}$, we have
\begin{align}	\label{eq: MIMO sensitivity bounds}
\|GK\|_2 \leq \Delta^{m,p}_2 G \leq |k|_2 \|G\|_2.
\end{align}
%where $K = \text{diag}(k_1, \ldots, k_m)$ and $|k|_2 = \left(\sum_{i=1}^m k_i^2\right)^{1/2}$. %, and $\|G\|_2$ is the $\mathcal H_2$ norm of $G$.
\end{thm}

\begin{proof}
We have
$
G(u-u') = \sum_{i=1}^m \alpha_i \, G \delta_{t_i} e_i,
$
and moreover $\|G\|^2_2 = \sum_{i=1}^m \|G \delta_{t_i} e_i\|_2^2$ by definition.
For the upper bound, we can write
\begin{align*}
\|G(u-u')\|_2 &= \left\| \sum_{i=1}^m \alpha_i \, G \delta_{t_i} e_i \right \|_2 \\
&\leq \sum_{i=1}^m | \alpha_i | \| G \delta_{t_i} e_i \|_2 \\
&\leq |k|_2 \left( \sum_{i=1}^m \| G \delta_{t_i} e_i \|^2_2 \right)^{1/2},
\end{align*}
where the last inequality results from the Cauchy-Schwarz inequality.

For the lower bound, let us first take $u' \equiv 0$. Then consider an adjacent signal $u$ with a single discrete impulse
of height $k_i$ at time $t_i$ on each input channel $i$, for $i=1,\ldots,m$, with $t_1<t_2<\ldots<t_m$. Let $\eta > 0$. 
Denote the ``columns'' of $G$ as $G_i$ for $i=1,\ldots,m$, i.e., $G u = \sum_{i=1}^m G_i u_i$.
%For $i=1,\ldots,m$, denote by $G_i$ the  system such that $
Since 
$\|G\|_2 < \infty$, $\|G_i u_i\|_2<\infty$, and hence $|(G_i u_i)_t| \to 0$ as $t \to \infty$. Hence by taking $t_{i+1}-t_i$
large enough for each $1 \leq i \leq m-1$, i.e., waiting for the effect of impulse $i$ on the output to be sufficiently small, 
we can choose the signal $u$ such that
\[
\|Gu\|_2^2 = \left\| \sum_{i=1}^m G_i u_i \right\|_2^2 \geq  \sum_{i=1}^m k_i^2 \|G \delta_{t_i} e_i\|_2^2-\eta.
\]
Since this is true for any $\eta>0$ and  $\|G \delta_{t_i} e_i\|_2^2 = \|G_i\|_2^2$, we get 
$(\Delta^{m,p}_2 G)^2 \geq \|GK\|_2^2 = \sum_{i=1}^m k_i^2 \|G_i\|^2$.
\end{proof}

Note that if $k_1= \ldots =k_m$, the upper bound on the sensitivity is $k_1 \|G\|_2 \sqrt{m}$. 
%The squared sensitivity, which is related to the variance of the required amount of privacy-preserving noise
%in an output perturbation scheme, scales then linearly with the number of inputs. 
We can compare this bound to the situation where $G$ is diagonal, in which case the sensitivity is exactly $k_1 \|G\|_2$ from 
Theorem \ref{thm: diagonal system sensitivity}. %Moreover, 
The following example shows that the upper bound of Theorem \ref{thm: bound MIMO general} cannot be improved 
for the general MISO or MIMO case.

\begin{exmp}
Consider the MISO system $G(z) = [G_{11}(z), \ldots, G_{1m}(z)]$, with $g_{1i} = \delta_{\tau_i}$ the impulse response of $G_{1i}$, 
for some times $\tau_1, \ldots, \tau_m$.
Then $\|G\|_2^2 = m$. Now let $u' \equiv 0$ and $u = \sum_{i=1}^m \delta_{t_i} e_i$, so that $u$ and $u'$ are adjacent, 
with $k_1=\ldots=k_m=1$, and moreover let us choose the times $t_i$ such that $\tau_i+t_i$ is a constant, i.e., take $t_i = \kappa-\tau_i$
for some $\kappa \geq \max_i\{{\tau_i\}}$. Then $Gu = \sum_{i=1}^m g_{1i} * u_i = m \delta_\kappa$, and so $\|Gu\|_2^2 = m^2$. This shows
that the upper bound of Theorem \ref{thm: bound MIMO general} is tight in this case. Note that this happens because all the events
of the signal $u$ influence the output at the same time. Indeed, if the times $\tau_i+t_i$ are all distinct, then we get $\|Gu\|_2^2 = m$.
\end{exmp}
% in other words, one approach could be to try to prevent this superposition, perhaps by adding carefully chosen delays?

 % new paper
% !TEX root =  ../dpEventFiltering.tex

\subsection{Exact solution for the MIMO Case}

For completeness, we give in this subsection an exact expression for the sensitivity of a MIMO filter. 
Let $G$ be a stable LTI system with $m$ inputs and $p$ outputs, and state space representation
\begin{align}	
x_{t+1} &= A x_t + B u_t \label{eq: ss representation} \\
y_t &=C x_t + D u_t, \nonumber
\end{align}
with $x_0 = 0$. Recall the definition of the observability Gramian $P_0$, which is the unique 
positive semi-definite solution of the equation
\[
A^T P_0 A - P_0 + C^T C = 0.
\]
Let $B_i, D_i$ be the $i^{th}$ column of the matrix $B$ and $D$ respectively, for $i=1,\ldots,m$.
Finally, define for $i,j \in \{1,\ldots,m\}$, $i \neq j$, 
and $\tau$ in $\mathbb Z$
\ifthenelse {\boolean{TwoColEq}} 
{
\begin{align}
&S_{ij}^{\tau} = \label{eq: cross terms} \\
&\begin{cases}
B_i^T (A^{\tau-1})^T C^T D_j + B_i^T (A^{\tau})^T P_0 B_j, & \text{if } \tau > 0 \\
D_i^T D_j + B_i^T P_0 B_j, & \text{if } \tau=0 \\
D_i^T C A^{|\tau|-1} B_j + B_i^T P_0 A^{|\tau|} B_j, & \text{if } \tau<0.
\end{cases} \nonumber
\end{align}
}{
\begin{align}  \label{eq: cross terms} 
S_{ij}^{\tau} =
\begin{cases}
%B_i^T (A^{\tau-1})^T C^T D_j + B_i^T (A^{\tau})^T P_0 B_j, & \text{if } \tau > 0 \\
D_j^T C A^{\tau-1} B_i + B_j^T P_0 A^{\tau} B_i, & \text{if } \tau > 0 \\
D_i^T D_j + B_i^T P_0 B_j, & \text{if } \tau=0 \\
D_i^T C A^{|\tau|-1} B_j + B_i^T P_0 A^{|\tau|} B_j, & \text{if } \tau<0.
\end{cases} 
\end{align}
}

\begin{thm}
Let $G$ be a stable LTI system with $m$ inputs and $p$ outputs, and state space representation (\ref{eq: ss representation}).
Then, for the adjacency relation (\ref{eq: adjacency relation}), we have
\begin{align}	\label{eq: sensitivity MIMO general}
(\Delta^{m,p}_{2} G)^2 = \|GK\|_2^2 + \sum_{\substack{i,j=1 \\ i \neq j}}^m k_i k_j \left(\sup_{t_i,t_j \in \mathbb N} \left| S^{t_j - t_i}_{ij} \right|\right).
\end{align}
\end{thm}

\begin{proof}
In view of (\ref{eq: adjacent signals}), we have
\[
\Delta_2^{m,p} G = \sup_{|\alpha_i| \leq k_i, t_i \geq 0} \left \| \sum_{i=1}^m \alpha_i G \delta_{t_i} e_i \right \|_2.
\]
For $y_i=G\delta_{t_i} e_i$ and $y = \sum_{i=1}^m \alpha_i y_i$, we have
\begin{align*}
\|y\|^2_2 &= \sum_{t=0}^\infty \left| \sum_{i=1}^m \alpha_i y_{i,t} \right|^2 \\
&= \sum_{t=0}^\infty \sum_{i=1}^m \alpha_i^2 |y_{i,t}|^2 + \sum_{t=0}^\infty \; \sum_{\substack{i,j=1 \\ i \neq j}}^m \alpha_i \alpha_j y_{i,t}^T y_{j,t} \\
&\leq \|GK\|_2^2 +  \sum_{\substack{i,j=1 \\ i \neq j}}^m k_i k_j \left| \sum_{t=0}^\infty y_{i,t}^T y_{j,t} \right|,
\end{align*}
where $K = \text{diag}(k_1,\ldots,k_m)$ and the bound can be attained by taking $\alpha_i \in \{-k_i,k_i\}$, depending 
on the sign of $S_{ij} := \sum_{t=0}^\infty y_{i,t}^T y_{j,t}$.

Next, we derive the more explicit expression for $S_{ij}$ given in the theorem.
First,
\[
y_{i,t} = \begin{cases}
0, & t < t_i, \\
D_i, & t = t_i \\
C A^{t-t_i-1} B_i, & t > t_i.
\end{cases}
\]
%Assume without loss of generality that $t_1 \leq t_2 \leq \ldots \leq t_m$. 
%\textcolor{red}{[Review calculation below]}
Then if $t_i=t_j$, we find that
\[
S_{ij} = D_i^T D_j + B_i^T P_0 B_j,
\]
with $P_0 = \sum_{t=0}^\infty (A^t)^T C^T C A^t$ the observability Gramian.
If $t_i < t_j$, then
\[
S_{ij} = B_i^T (A^{t_j-t_i-1})^T C^T D_j + B_i^T (A^{t_j-t_i})^T P_0 B_j,
\]
which corresponds to the first case in (\ref{eq: cross terms}).
The case $t_i>t_j$ is symmetric.
\end{proof}

\subsection{Discussion}

In (\ref{eq: sensitivity MIMO general}), the maximization over inter-event times $t_i-t_j$ still needs 
to be performed and depends on the parameters of the specific system $G$.
%, since the result of the maximization of inter-event times $t_i-t_j$ depends on these parameters. 
This result could be used to evaluate carefully the amount of noise necessary in an output perturbation
mechanism, but unfortunately it seems too unwieldy at this point to be used in more advanced mechanism optimization 
schemes, such as the ones discussed in the next sections.

Still, the expression (\ref{eq: sensitivity MIMO general}) provides some intuition about the way the system 
dynamics influence its sensitivity. In particular, the second term in (\ref{eq: sensitivity MIMO general}) 
can give insight into the gap between the sensitivity and the lower bound in (\ref{eq: MIMO sensitivity bounds}).
Note from the expression of $S_{ij}^{\tau}$ in (\ref{eq: cross terms}) that one way to decrease
the sensitivity of $G$ is to increase sufficiently the required time $|t_i-t_j|$ between the events contributed by a 
single user, in order for $\|A^{|t_i-t_j|}\|$ to be small enough. Hence, a lower bound on inter-event
times in different streams could be introduced in the adjacency relation to reduce a system's sensitivity. 
This would weaken the differential privacy guarantee but help in the design of mechanisms with better performance. 
Another possibility would be to have a privacy-preserving mechanism simply ignore events from a given user 
as long as the lower bound on inter-event times is not reached.
%, provided these events can be easily and immediately identified by the system.

% !TEX root =  ../dpEventFiltering.tex

\section{Zero-Forcing MIMO Mechanisms}		\label{section: ZFM}

%\textcolor{red}{[Stability requirements of the pre-filters $G$ in the thms $6$ and $7$ needs to be clarified. Probably
%necessary to write the sensibility and mse in the first place, see the thms above.
%Paley-Wiener condition, etc.: see Kailath-Sayed-Hassibi book p. 186 for example. Do we need inverse filter to
%have finite H2 norm as well? See KSH chapter 6 on spectral factorization.
%]}

Using the sensitivity calculations of Section \ref{section: sensitivity}, we can now design differentially private 
mechanisms to approximate a given filter $F$, as discussed in Section \ref{section: generic scenario}. The 
mechanisms described in this section generalize to the MIMO case some ideas introduced in \cite{LeNy_DP_CDC12}. 
%Indeed, the general approximation architecture considered, described on Fig. \ref{fig: approximation setup}, is the same as for the SISO case.
The general approximation architecture considered is described on Fig. \ref{fig: approximation setup}.
On this figure, the system $H$ is of the form $H=FL$, with $L$ a left inverse of the pre-filter $G$. 
We call the resulting mechanisms Zero-Forcing Equalization (ZFE) mechanisms. 
Our goal is to design $G$ (and hence, $H$) so that the Mean Square Error (MSE) between $y$ and $\hat y$ 
on Fig. \ref{fig: approximation setup} is minimized. 
In order to obtain a differentially private signal $v$, we introduce a Gaussian white noise signal $w$ with variance 
proportional to the sensitivity of the filter $G$. %, following Theorem \ref{eq: basic DP mechanism}. 
It was shown in \cite{LeNy_DP_CDC12} for the SISO case that this setup can allow significant performance improvements compared to the
output-perturbation mechanism. Note that the latter is recovered when $G=F$ and $H$ is the identity.

\ifthenelse {\boolean{TwoColEq}} 
{
\begin{figure}
\includegraphics[width=\linewidth]{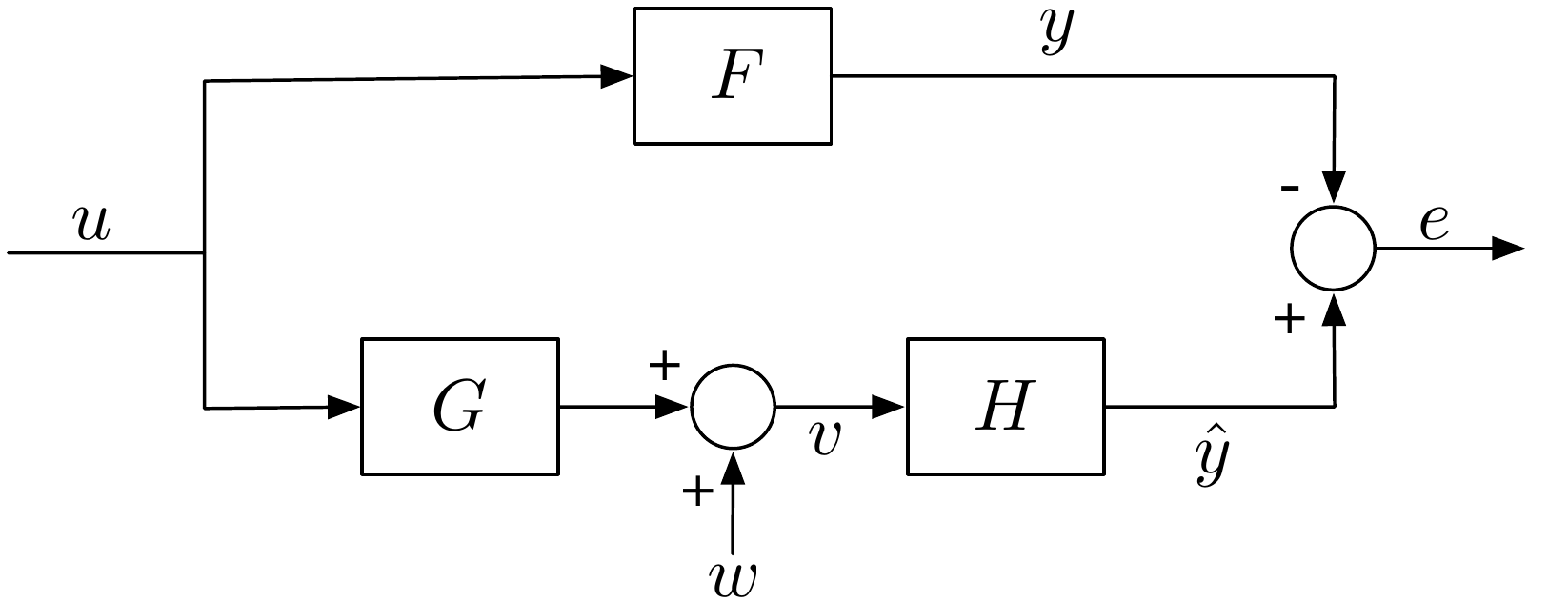}
\caption{Approximation setup for differentially private filtering. $w$ is a noise signal guaranteeing
that $v$ is a differentially private signal. The signal $\hat y$ is differentially private no matter what
the system $H$ is, see \cite[Theorem 1]{LeNyDP2012_journalVersion}.}
\label{fig: approximation setup}
\end{figure}
}{
\begin{figure}
\centering
\includegraphics[width=0.7\linewidth]{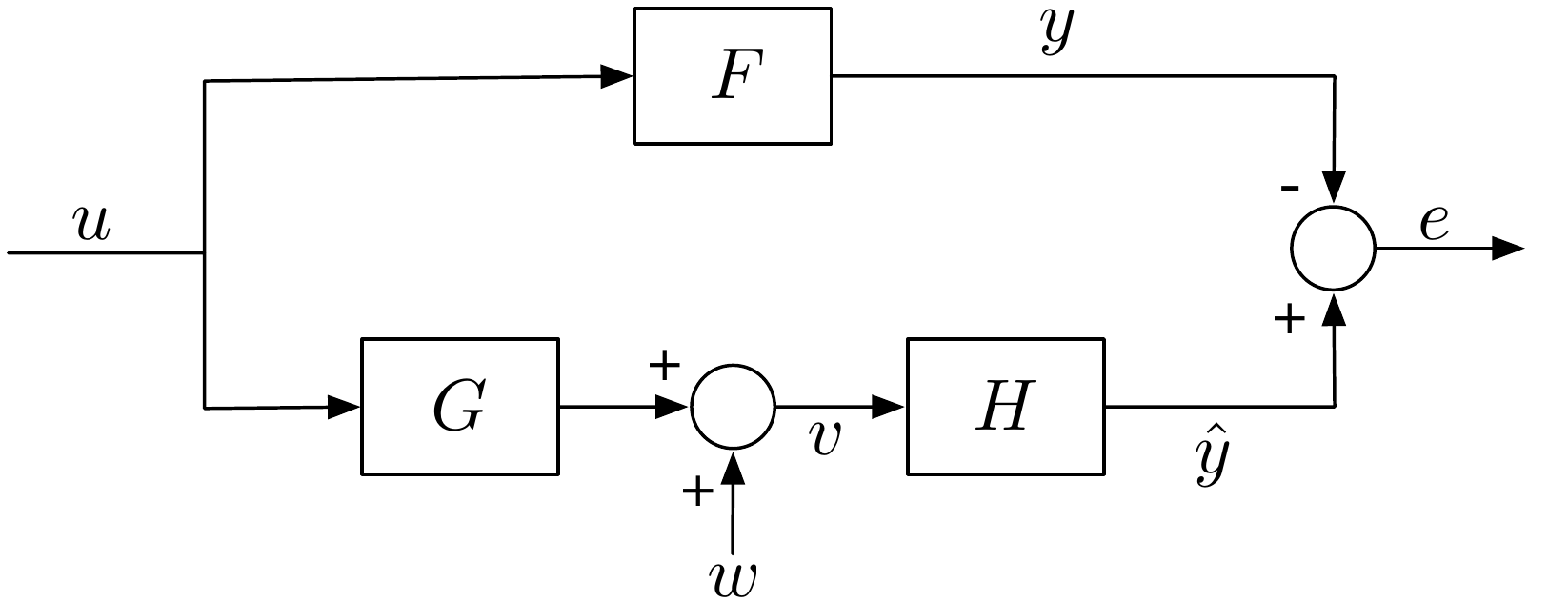}
\caption{Approximation setup for differentially private filtering. The signal $w$ is a noise signal guaranteeing
that $v$ is a differentially private signal. The signal $\hat y$ is differentially private no matter what
the system $H$ is, see \cite[Theorem 1]{LeNyDP2012_journalVersion}.}
\label{fig: approximation setup}
\end{figure}
}

\subsection{SIMO system approximation}

First, let us assume that $F$ on Fig. \ref{fig: approximation setup} is a SIMO filter, with $p$ outputs. 
Consider a first stage $G(z)=\text{col}(G_{1}(z),\ldots,G_{q}(z))$ taking the input signal $u$ and producing 
$q$ intermediate outputs that must be perturbed. 
%The second stage is taken to be $H=FG^{\dagger}$, with $G^\dagger$ a left-inverse of $G$, i.e.,
%satisfying 
The second stage is taken to be $H=FL$, with $L(z)=[L_{1}(z),\ldots,L_{q}(z)]$ a left-inverse of $G$, i.e., satisfying 
\[
%L_{1}(z) G_{1}(z) + \ldots L_{q}(z) G_{q}(z) = 1.
\sum_{i=1}^q L_i(z) G_i(z) = 1.
\]
Let us also define the transfer functions $M_i$, $i=1,\ldots,q,$ such that $M_i(z) = L_i(z^{-1})$, hence
$M_i(e^{j \omega}) = L_i(e^{j \omega})^*$, and thus in particular
\begin{align}
&|M_i(e^{j \omega})|^2 = |L_i(e^{j \omega})|^2, i=1,\ldots,q,  \label{eq: L property 1} \\
%&M_1(e^{j \omega})^* G_1(e^{j \omega}) + \ldots + M_q(e^{j \omega})^* G_q(e^{j \omega}) = 1. \label{eq: L property 2}
\text{and } & \sum_{i=1}^q M_i(e^{j \omega})^* G_i(e^{j \omega}) = 1.   \label{eq: L property 2}
\end{align}

From Theorem \ref{thm: SIMO sensitivity}, the sensitivity of the first stage 
for input signals that are adjacent according to (\ref{eq: adjacency relation}) is
$k_1 \|G\|_2$. 
%\begin{align*}
%(\Delta_2 G)^2 &= \|G\|_2^2 = \sum_{i=1}^k \| G_{i1} \|_2^2 \\
%&= \frac{1}{2 \pi} \int_{0}^{2\pi} \sum_{i=1}^k |G_{i1}(e^{j \omega})|^2 d \omega.
%\end{align*}
Hence, according to Theorem \ref{eq: basic DP mechanism}, adding a white Gaussian noise $w$ to the
output of $G$ with covariance matrix $k_1^2 \kappa_{\delta,\epsilon}^2 \|G\|_2^2 I_q$
is sufficient to ensure that the signal $v$ on Fig. \ref{fig: approximation setup} is differentially
private. The MSE for this mechanism can be expressed as
\begin{align*}
e^{ZFE}_{mse}(G) &= \lim_{T \to \infty} \frac{1}{T} \sum_{t=0}^\infty \mathbb E \left[ | (Fu)_t - (FLGu)_t - (FLw)_t |_2^2 \right] \\
e^{ZFE}_{mse}(G) &= \lim_{T \to \infty} \frac{1}{T} \sum_{t=0}^\infty \mathbb E \left[ | (FLw)_t |_2^2 \right] \\
e^{ZFE}_{mse}(G) &= k_1^2 \kappa_{\delta,\epsilon}^2 \|G\|_2^2 \|FL\|_2^2.
\end{align*}
We are thus led to consider the minimization of $\|FL\|_2^2  \|G\|_2^2$ over the pre-filters $G$. 
%Recall in the following calculation that $| \cdot |$ is also used to denote the Euclidean norm in $\mathbb C^d$, for any $d$. 
We have
\ifthenelse {\boolean{TwoColEq}} 
{
\begin{align*}
\|&FL\|_2^2  \|G\|_2^2 \\
=&\frac{1}{2\pi} \int_{-\pi}^\pi \Tr (L^*(e^{j \omega}) F^*(e^{j \omega}) F(e^{j \omega}) L(e^{j \omega})) d \omega \times \\
&\frac{1}{2 \pi} \int_{-\pi}^\pi \Tr (G^*(e^{j \omega}) G(e^{j \omega})) d \omega \\
=& \frac{1}{2\pi} \int_{-\pi}^\pi |F(e^{j \omega})|_2^2 \; |L(e^{j \omega})|^2 d \omega 
\times \frac{1}{2\pi} \int_{-\pi}^\pi |G(e^{j \omega})|^2 d \omega \\
=& \frac{1}{2\pi} \int_{-\pi}^\pi |F(e^{j \omega})|_2^2 \; |M(e^{j \omega})|^2 d \omega 
\times \frac{1}{2\pi} \int_{-\pi}^\pi |G(e^{j \omega})|^2 d \omega,
\end{align*}
}{
\begin{align*}
\|FL\|_2^2  \|G\|_2^2 
&=\frac{1}{2\pi} \int_{-\pi}^\pi \Tr (L^*(e^{j \omega}) F^*(e^{j \omega}) F(e^{j \omega}) L(e^{j \omega})) d \omega \times 
\frac{1}{2 \pi} \int_{-\pi}^\pi \Tr (G^*(e^{j \omega}) G(e^{j \omega})) d \omega \\
&= \frac{1}{2\pi} \int_{-\pi}^\pi |F(e^{j \omega})|_2^2 \; |L(e^{j \omega})|_2^2 d \omega 
\times \frac{1}{2\pi} \int_{-\pi}^\pi |G(e^{j \omega})|_2^2 d \omega \\
&= \frac{1}{2\pi} \int_{-\pi}^\pi |F(e^{j \omega})|_2^2 \; |M(e^{j \omega})|_2^2 d \omega 
\times \frac{1}{2\pi} \int_{-\pi}^\pi |G(e^{j \omega})|_2^2 d \omega,
\end{align*}
}
where in the last equality we used (\ref{eq: L property 1}). Now consider the following inner product on the space of
$2\pi$-periodic functions with values in $\mathbb C^q$ 
\[
\langle f, g \rangle = \frac{1}{2 \pi} \int_{-\pi}^\pi f(e^{j \omega})^* g(e^{j \omega}) d \omega.
\]
By the Cauchy-Schwarz inequality for this inner product applied to the functions 
$\omega \mapsto |F(e^{j \omega})|_2 M(e^{j \omega})$ and 
$\omega \mapsto G(e^{j \omega})$, we obtain the following bound
\ifthenelse {\boolean{TwoColEq}} 
{
\begin{align*}
&\|FL\|_2^2  \|G\|_2^2 \geq \\
&\left( \frac{1}{2 \pi} \int_{-\pi}^\pi |F(e^{j \omega})|_2 \sum_{i=1}^q M_i^*(e^{j \omega}) G_i(e^{j \omega}) d \omega \right)^2,
\end{align*}
}{
\begin{align*}
\|FL\|_2^2  \|G\|_2^2 \geq 
\left( \frac{1}{2 \pi} \int_{-\pi}^\pi |F(e^{j \omega})|_2 \sum_{i=1}^q M_i^*(e^{j \omega}) G_i(e^{j \omega}) d \omega \right)^2,
\end{align*}
}
i.e., using (\ref{eq: L property 2}),
\[
\|FL\|_2^2  \|G\|_2^2 \geq \left( \frac{1}{2 \pi} \int_{-\pi}^\pi |F(e^{j \omega})|_2 \, d \omega \right)^2.
\]
Moreover, the two sides in the Cauchy-Schwarz inequality are equal, i.e., the bound is attained, if
\[
|F(e^{j \omega})|_2 M(e^{j \omega}) = G(e^{j \omega}).
\]
Note that this condition does not depend on $q$. Hence we can simply take $q=1$, and $L(z) = 1/G(z)$,
to get
\begin{align}
|F(e^{j \omega})|_2 \; L^*(e^{j \omega}) = G(e^{j \omega}) \nonumber \\
\text{i.e., } |G(e^{j \omega})|^2 = |F(e^{j \omega})|_2.	\label{eq: SIMO optimization result}
\end{align}
Finding $G$ SISO satisfying (\ref{eq: SIMO optimization result}) is a spectral factorization problem.
We can choose $G$ stable and minimum phase, so that its inverse $L$ is also stable. The following
theorem summarizes the preceding discussion and generalizes \cite[Theorem 8]{LeNyDP2012_journalVersion}.

\begin{thm}	\label{thm: error for ZFE mechanism}
Let $F$ be a SIMO LTI system with $\|F\|_2 < \infty$.
For any stable LTI system $G$, %$\mathcal G_1$, 
\begin{equation}	\label{eq: ZFE bound in thm}
%e^{ZFE}_{mse} (\mathcal G_1) 
e^{ZFE}_{mse}(G)
\geq k_1^2 \kappa_{\delta,\epsilon}^2 \left(\frac{1}{2 \pi} \int_{-\pi}^\pi |F(e^{j\omega})|_2 \; d \omega \right)^2.
\end{equation}
%where $|F(e^{j\omega})| = \left( \sum_{i=1}^p |F_{i1}(e^{j\omega})|^2 \right)^{1/2}$.
If moreover $F$ satisfies the Paley-Wiener condition 
$\frac{1}{2 \pi} \int_{-\pi}^\pi \ln |F(e^{j \omega})|_2 \; \, d \omega > - \infty$, 
this lower bound on the mean square error of the ZFE mechanism can be attained by some stable minimum 
phase SISO system $G$ such that $|G(e^{j \omega})|^2 = |F(e^{j \omega})|_2$, for almost every 
$\omega \in [-\pi,\pi)$.
\end{thm}

\begin{proof}
The main argument for the proof was given before the theorem. 
Since $|F(e^{j \omega})|_2$ is a nonnegative function on the unit circle, if it satisfies the Paley-Wiener condition, 
it has indeed a minimum phase spectral factor $G$ satisfying (\ref{eq: SIMO optimization result}) almost everywhere \cite[p. 242]{Poor94_SPbook}.
%and thus the performance bound can be attained.
%\textcolor{red}{[Currently here. To finish]}.
\end{proof}

\subsection{MIMO system approximation}		\label{section: diagonal pre-filter optimization}

Let us now assume that $F$ has $m>1$ inputs. We write $F(z) = [F_1(z),\ldots,F_m(z)]$,
with $F_i$ a $p \times 1$ transfer matrix.
In this case, in view of the complicated expression (\ref{eq: sensitivity MIMO general}) 
for the sensitivity of a general MIMO filter, we only provide a subpotimal ZFE mechanism, 
together with a comparison between the performance of our mechanism and the optimal ZFE mechanism.
The idea is to restrict our attention to pre-filters $G$ that are $m \times m$ and diagonal, for
which the sensitivity is given in Theorem \ref{thm: diagonal system sensitivity}. The problem
of optimizing the diagonal pre-filters, using the architecture depicted on Fig. \ref{fig: approximation setup MIMO},
 can in fact be seen as designing $m$ SIMO mechanisms.

\ifthenelse {\boolean{TwoColEq}} 
{
\begin{figure}
\includegraphics[width=\linewidth]{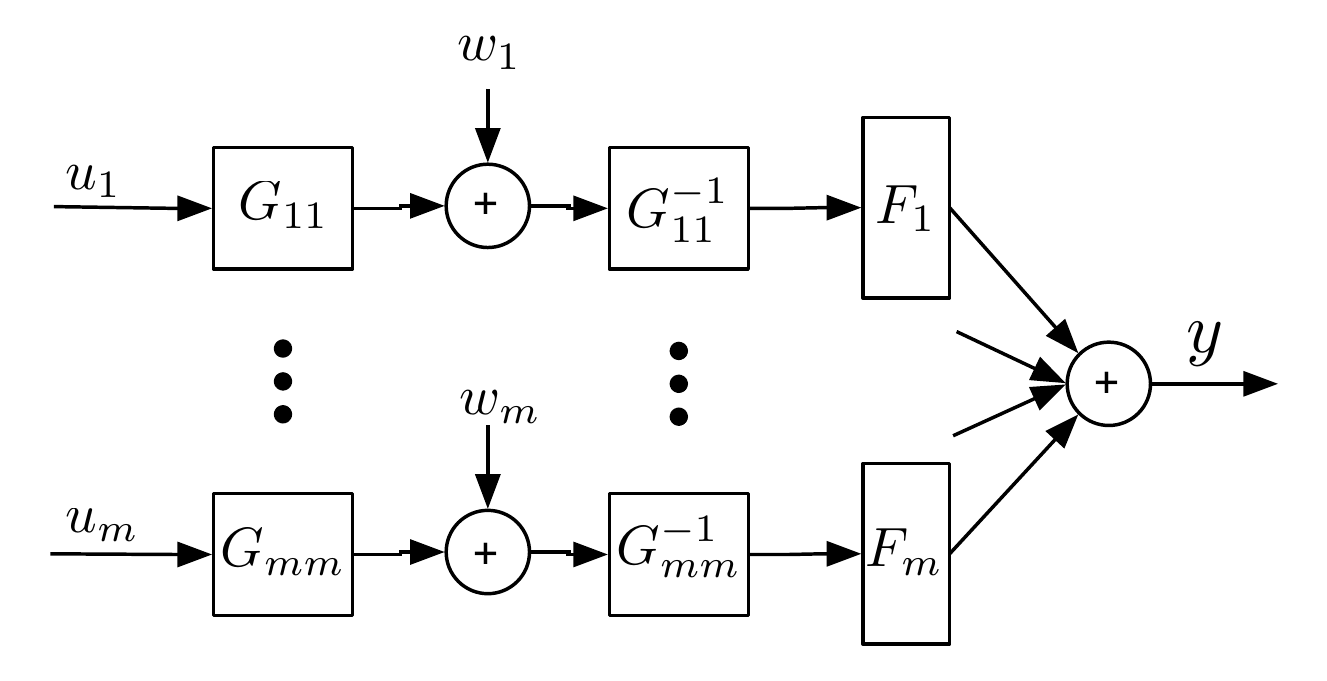}
\caption{(Suboptimal) ZFE mechanism for a MIMO system $Fu = \sum_{i=1}^m F_i u_i$, and a diagonal pre-filter 
$G(z) = \text{diag}(G_{11}(z),\ldots,G_{mm}(z))$. Here $F_i(z)$ is a $p \times 1$ transfer matrix, for
$i=1,\ldots,m$. The signal $w$ is a white Gaussian noise with covariance matrix 
$\kappa_{\delta,\epsilon}^2 \|KG\|_2^2 I_m$.}
\label{fig: approximation setup MIMO}
\end{figure}
}{
\begin{figure}
\centering
\includegraphics[width=0.6\linewidth]{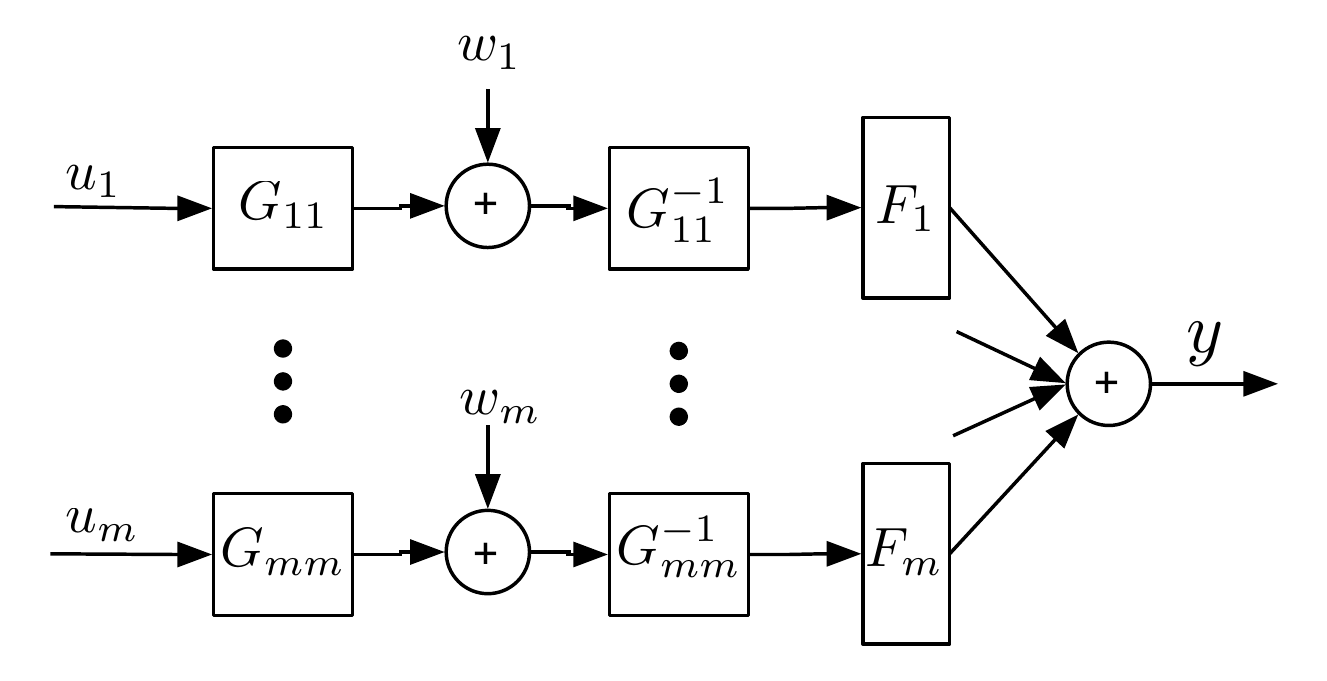}
\caption{(Suboptimal) ZFE mechanism for a MIMO system $Fu = \sum_{i=1}^m F_i u_i$, and a diagonal pre-filter 
$G(z) = \text{diag}(G_{11}(z),\ldots,G_{mm}(z))$. Here $F_i(z)$ is a $p \times 1$ transfer matrix, for
$i=1,\ldots,m$. The signal $w$ is a white Gaussian noise with covariance matrix 
$\kappa_{\delta,\epsilon}^2 \|KG\|_2^2 I_m$.}
\label{fig: approximation setup MIMO}
\end{figure}
}

\subsubsection{Diagonal Pre-filter Optimization}

If $G$ is diagonal, then according to Theorem \ref{thm: diagonal system sensitivity} its squared sensitivity 
is $(\Delta^{m,m}_2 G)^2 = \|KG\|_2^2 = \sum_{i=1}^m \|k_i G_{ii}\|_2^2$, with $K = \text{diag}(k_1,\ldots,k_m)$.
%Assume moreover that each $G_{ii}(z)$ is invertible. 
Following the same reasoning as in the previous subsection, the MSE for this mechanism can be expressed as
\begin{align}	\label{eq: performance ZFE}
e^{ZFE}_{mse}(G) &= \kappa_{\delta,\epsilon}^2 \|KG\|_2^2 \|FG^{-1}\|_2^2,
\end{align}
with $G^{-1}(z) = \text{diag}(G_{11}(z)^{-1},\ldots,G_{mm}(z)^{-1})$. Now remark that
\[
\|FG^{-1}\|_2^2 = \frac{1}{2 \pi} \int_{-\pi}^\pi \sum_{i=1}^m \frac{|F_i(e^{j \omega})|_2^2}{|G_{ii}(e^{j \omega})|^2} d \omega.
\]
Hence from the Cauchy-Schwarz inequality again, we obtain the lower bound
\begin{align*}
e^{ZFE}_{mse}(G) &\geq \kappa_{\delta,\epsilon}^2 \left( \frac{1}{2 \pi} \int_{-\pi}^\pi \sum_{i=1}^m 
\frac{|F_i(e^{j \omega})|_2}{|G_{ii}(e^{j \omega})|} |k_i G_{ii}(e^{j \omega})| d \omega \right)^2 \\
e^{ZFE}_{mse}(G) &\geq \kappa_{\delta,\epsilon}^2 \left( \frac{1}{2 \pi} \int_{-\pi}^\pi \sum_{i=1}^m k_i |F_i(e^{j \omega})|_2 \; d \omega \right)^2,
\end{align*}
and this bound is attained if
\begin{align*}
k_i |G_{ii}(e^{j \omega})| &= \frac{|F_i(e^{j \omega})|_2}{|G_{ii}(e^{j \omega})|},  \\
\text{i.e. } \; k_i |G_{ii}(e^{j \omega})|^2 &= |F_i(e^{j \omega})|_2, \; i=1,\ldots,m.
\end{align*}
In other words, the best diagonal pre-filter for the MIMO ZFE mechanism can be obtained from 
$m$ spectral factorizations of the functions $\omega \mapsto \frac{1}{k_i} |F_i(e^{j \omega})|_2$,
$i=1,\ldots,m$. 

\begin{thm}	\label{thm: error for ZFE mechanism MIMO diagonal}
Let $F=[F_1,\ldots,F_m]$ be a MIMO LTI system with $\|F\|_2 < \infty$.
We have, for any stable diagonal filter $G(z)=\text{diag}(G_{11}(z),\ldots,G_{mm}(z))$, %$\mathcal G_1$, 
\begin{equation}	\label{eq: ZFE bound in thm - diagonal}
%e^{ZFE}_{mse} (\mathcal G_1) 
e^{ZFE}_{mse}(G)
\geq \kappa_{\delta,\epsilon}^2 \left( \frac{1}{2 \pi} \int_{-\pi}^\pi \sum_{i=1}^m k_i |F_i(e^{j \omega})|_2 \; d \omega \right)^2.
\end{equation}
If moreover each $F_i$ satisfies the Paley-Wiener condition 
$\frac{1}{2 \pi} \int_{-\pi}^\pi \ln |F_i(e^{j \omega})|_2 \, d \omega > - \infty$, 
this lower bound on the mean-squared error of the ZFE mechanism can be attained by some stable minimum 
phase systems $G_{ii}$ such that $|G_{ii}(e^{j \omega})|^2 = \frac{1}{k_i} |F_i(e^{j \omega})|_2$, for almost every 
$\omega \in [-\pi,\pi)$.
\end{thm}

\begin{rem}
Note that the integrand on the right-hand side of (\ref{eq: ZFE bound in thm - diagonal}) can be written
\[
\sum_{i=1}^m k_i |F_i(e^{j \omega})|_2 := \| F(e^{j \omega}) K \|_{2,1},
\]
where $K = \text{diag}(k_1,\ldots,k_m)$ as usual and $\| \cdot \|_{2,1}$ is the so-called $L_{2,1}$ or $R_1$ matrix norm,
and appears in \cite{Ding:ML06:R1PCA} for example.
\end{rem}

\subsubsection{Comparison with Non-Diagonal Pre-filters}

For $F$ a general MIMO system, it is possible that we could achieve a better performance with a ZFE mechanism
where $G$ is not diagonal, i.e., by combining the inputs before adding the privacy-preserving noise. 
To provide a better understanding of how much could potentially be gained by carrying out this more involved 
optimization over general pre-filters $G$ rather than just diagonal pre-filters, the following Theorem provides a
a lower bound on the MSE achievable by \emph{any} ZFE mechanism. 
%Here we provide a lower bound on the MSE that one could potentially achieve by carrying out this more involved optimization over
%general pre-filters $G$ rather than just diagonal pre-filters. 
%To simplify the discussion, we assume $k_1=\ldots=k_m=1$.
%
%To simplify the discussion, we assume that $k_1, \ldots k_m$ are strictly positive, and 
%We denote as usual $K = \text{diag}(k_1,\ldots,k_m)$.

\begin{thm}	\label{thm: error for ZFE mechanism MIMO general}
Let $F=[F_1,\ldots,F_m]$ be a MIMO LTI system with $\|F\|_2 < \infty$.
We have, for any $m \times m$ stable filter $G(z)$, %$\mathcal G_1$, 
with left inverse $L$ so that $L(z) G(z) = I$,
\begin{equation}	\label{eq: ZFE bound in thm - diagonal}
e^{ZFE}_{mse}(G) \geq
\kappa_{\delta,\epsilon}^2 \left(  \frac{1}{2 \pi} \int_{- \pi}^{\pi}  \|F(e^{j \omega}) K\|_* d \omega \right)^2,
\end{equation}
where $\|F(e^{j \omega}) K \|_*$ denotes the nuclear norm of the matrix 
$F(e^{j \omega}) K$ (sum of singular values).
\end{thm}

The lower bound \eqref{lower bound: general case} on the achievable MSE with a general pre-filter in a ZFE mechanism 
should be compared to the performance (\ref{eq: ZFE bound in thm - diagonal}) that can actually be achieved with diagonal pre-filters.
Note that these bounds co\"incide for $m=1$. For $m > 1$, the gap depends on the difference between the integrals of the
$L_{2,1}$ norm and the nuclear norm of $F(e^{j \omega}) K$.

\begin{proof}
%Hence, consider a general $m \times m$ pre-filter $G$ with left inverse $L$, so that $L(z) G(z) = I$. 
We denote as usual $K = \text{diag}(k_1,\ldots,k_m)$, and define $\tilde G = GK$ and $\tilde L = K^{-1} L$, 
so that we again have $\tilde L \tilde G = I$. Let $\tilde F = F K$.
% since we are looking at a lower bound for a family of mechanisms, we can restrict our attention to filters
% G for which L exists.
With the lower bound of Theorem \ref{thm: bound MIMO general}, designing a 
ZFE mechanism based on sensitivity as above would require adding a noise with variance at least 
$\kappa_{\delta,\epsilon}^2 \| \tilde G \|_2^2$.
This would lead to an MSE at least equal to $\kappa_{\delta,\epsilon}^2 \| \tilde G\|_2^2 \| \tilde F \tilde L\|_2^2$. Now note that
\begin{align*}
\|\tilde F \tilde L\|_2^2 &= \frac{1}{2 \pi} \int_{- \pi}^{\pi} \Tr (\tilde F(e^{j \omega}) \tilde L(e^{j \omega}) 
\tilde L(e^{j \omega})^* \tilde F(e^{j \omega})^*) d \omega \\
&= \frac{1}{2 \pi} \int_{- \pi}^{\pi} \Tr (\tilde F(e^{j \omega})^* \tilde F(e^{j \omega}) \tilde L(e^{j \omega}) \tilde L(e^{j \omega})^*) d \omega \\
&= \frac{1}{2 \pi} \int_{- \pi}^{\pi} \Tr (A(e^{j \omega})^2 \tilde L(e^{j \omega}) \tilde L(e^{j \omega})^*) d \omega \\
&= \frac{1}{2 \pi} \int_{- \pi}^{\pi} \Tr (A(e^{j \omega}) \tilde L(e^{j \omega}) \tilde L(e^{j \omega})^* A(e^{j \omega})) d \omega,
\end{align*}
where for all $\omega$, $A(e^{j \omega})$ is the unique Hermitian positive-semidefinite square root 
of $\tilde F(e^{j \omega})^* \tilde F(e^{j \omega})$, i.e., $A(e^{j \omega})^2 = \tilde F(e^{j \omega})^* \tilde F(e^{j \omega})$.
%\todo{perhaps give ref. to Horn and Johnson} 
% Reference on pos-sem. square root (and uniqueness):
%http://www.caam.rice.edu/~caam440/chapter2.pdf
%Horn and Johnson
% Here F*F is hermitian, pos. sem. def.
%
Then, once again from the Cauchy-Schwarz inequality, now for the inner product
$
\langle M, N \rangle = \frac{1}{2 \pi} \int_{-\pi}^\pi \Tr (M(e^{j \omega})^* N(e^{j \omega})) d \omega,
$
\begin{align}
\|GK\|_2^2 \; \|F L\|_2^2 &= \|\tilde G \|_2^2 \|\tilde F \tilde L\|_2^2 = \left[ \frac{1}{2 \pi} \int_{- \pi}^{\pi} \Tr (\tilde G(e^{j \omega})^* \tilde G(e^{j \omega})) d \omega \right] \nonumber \\
\times \Bigg[ \frac{1}{2 \pi} &\int_{- \pi}^{\pi} \Tr (A(e^{j \omega}) \tilde L(e^{j \omega}) \tilde L(e^{j \omega})^* A(e^{j \omega})) d \omega \Bigg] \nonumber \\
%\text{hence, } \; \|GK\|_2^2 \; \|F L\|_2^2 
&\geq \left(  \frac{1}{2 \pi} \int_{- \pi}^{\pi}  \Tr (A(e^{j \omega}) \tilde L(e^{j \omega}) \tilde G(e^{j \omega})) d \omega \right)^2 \nonumber \\
\text{and so } \; e^{ZFE}_{mse}(G) &\geq
\kappa_{\delta,\epsilon}^2 \left(  \frac{1}{2 \pi} \int_{- \pi}^{\pi}  \|F(e^{j \omega}) K\|_* d \omega \right)^2,
%\kappa_{\delta,\epsilon}^2 \left(  \frac{1}{2 \pi} \int_{- \pi}^{\pi}  \Tr (A(e^{j \omega})) d \omega \right)^2
    \label{lower bound: general case} 
%\|G\|_2^2 \|F L\|_2^2 &\geq \left(  \frac{1}{2 \pi} \int_{- \pi}^{\pi}  \Tr (A(e^{j \omega})) d \omega \right)^2 \\
\end{align}
where $\|F(e^{j \omega}) K \|_* = \Tr (A(e^{j \omega}))$ denotes the nuclear norm of the matrix $F(e^{j \omega}) K$. 
% (sum of singular values).
%The lower bound \eqref{lower bound: general case} on the achievable MSE with a general pre-filter in a ZFE mechanism 
%should be compared to the performance (\ref{eq: ZFE bound in thm - diagonal}) that we obtained with diagonal pre-filters.
%Note that these bounds co\"incide for $m=1$. For $m > 1$, the gap depends on the difference between the integrals of the
%$L_{2,1}$ norm and the nuclear norm of $F(e^{j \omega}) K$.
\end{proof}

\section{Application to Privacy-Preserving Estimation of Building Occupancy}    \label{section: building monitoring}

%\subsection{Forecasting Building Occupancy}	\label{section: building monitoring}

%\begin{figure}
%\includegraphics[width=\linewidth]{graphics/map1.pdf}
%\caption{Floorplan of the part of the MERL building used for the sensor network experiment in \cite{MERLdataset_2007}.
%The shaded areas are hallways, lobbies and meeting rooms equipped with more than 200 motion detectors, placed
%a few meters apart and recording events roughly every seconds.}
%\label{fig: MERL floorplan}
%\end{figure}

\begin{figure}
\includegraphics[width=\linewidth]{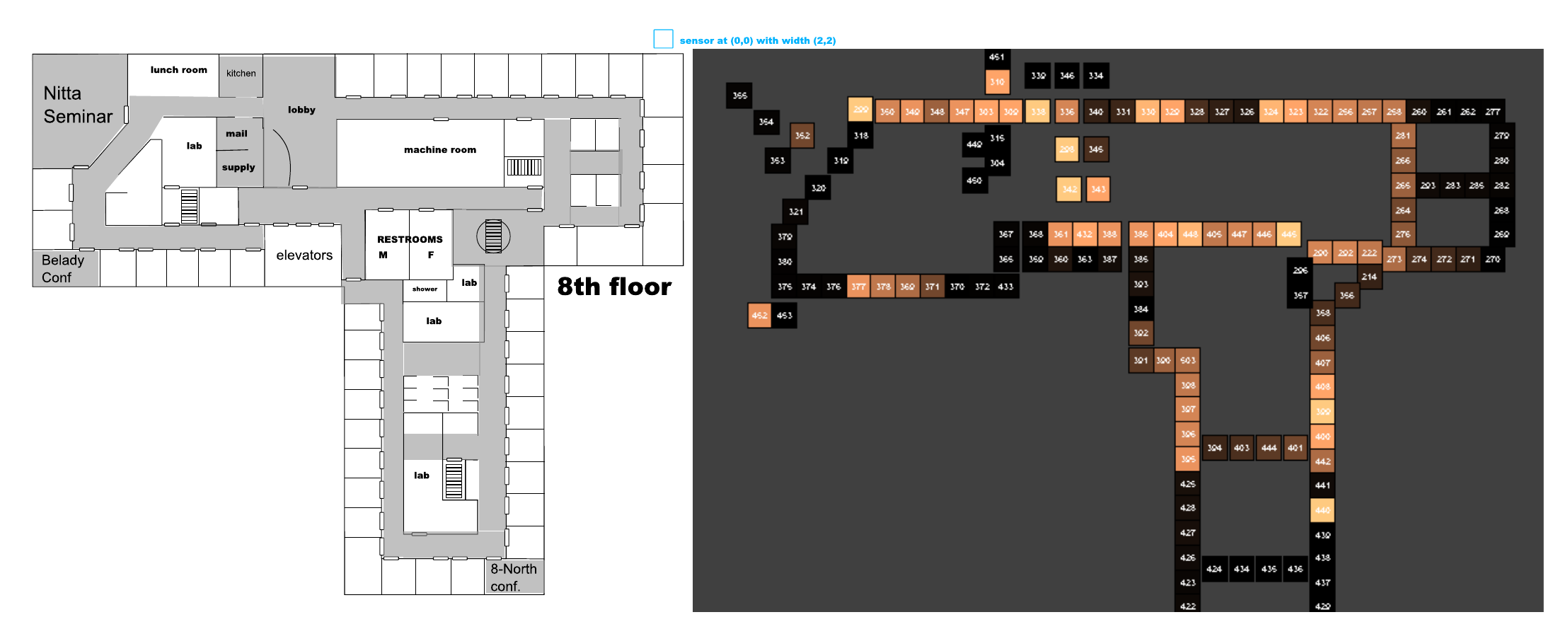}
\caption{Left: plan of one of the two floors in the MERL building used for the sensor network experiment  \cite{MERLdataset_2007}. 
The shaded areas are hallways, lobbies and meeting rooms equipped with binary motion detection sensors, placed a few meters 
apart and recording events roughly every second. Right: a figure taken from \cite{MERL:07:sensorNetVisu} shows a visualisation of a crowd 
movement during a fire drill.}
\label{fig: MERL floorplan}
\end{figure}

In this section we illustrate the design process and the performance of the ZFE mechanism 
in the context of an application to filtering and forecasting occupancy-related events in an office building equipped 
with motion detection sensors. 
As mentioned in Section \ref{section: DP def}, such sensor networks raise privacy concerns since some occupants could potentially 
be tracked individually from the published information, especially when it is correlated with public information such as the location 
of their office. Since the amount of private information leakage depends on the output signals the system aims to generates, we
adjust the privacy-preserving noise level based on the filter specification using the ZFE mechanism.
As an example, we simulate the outputs of a $3 \times 15$ MIMO filter processing input signals collected during a sensor network 
experiment carried out at the Mitsubishi Electric Research Laboratories (MERL) and described in \cite{MERLdataset_2007} and 
on Fig. \ref{fig: MERL floorplan}. We refer the reader to \cite{MERL:07:sensorNetVisu} for examples of identification of individual 
trajectories from this dataset.

The original dataset contains the traces of $213$ sensors placed a few meters apart and spread over two floors of a building, 
where each sensor recorded roughly every second and with millisecond accuracy over a year the exact times at which they detected 
some motion. 
For illustration purposes we downsampled the dataset in space and time, summing all the events recorded by several sufficiently 
close sensors over $3$ minute intervals. 
From this step, we obtained $15$ input signals $u_i$, $i=1,\ldots,15$, corresponding to $15$ spatial zones 
(each zone covered by cluster of about $14$ sensors), with a discrete-time period corresponding to $3$ minutes, and $u_{i,t}$ 
being the number of events detected by all the sensors in zone $i$ during period $t$. Let us assume say that during a given
discrete-time period, a single individual can activate at most $4$ sensors in any group, hence $k_i=4$ for $1 \leq i \leq 15$.
%If moreover we assume that any individual travels through at most $5$ zones, then we could add a constraint $\sum_{i=1}^{10} k_i \leq 10$.
Moreover, we need to assume that a single individual only activates the sensors in a given zone once over the time interval for which we 
wish to provide differential privacy. Section \ref{section: DP def} discussed how to relax this requirement by splitting the input data
into successive time windows and creating additional input channels. 
%Another possibility with more sophisticated sensors would be to recognize a person a do not record the corresponding events

%As an illustrative example, we could be interested in a system computing simultaneously and in real-time the following three outputs:
For our example, we consider computing simultaneously and in real-time the following three outputs from the $15$ input signals
\begin{equation}	\label{eq: MERL desired F}
\begin{bmatrix} y_1 \\ y_2 \\ y_3 \end{bmatrix} = 
\begin{bmatrix} f_1(z) \mathbf{1}_{1 \times 5} \;\;\;\; \mathbf{0}_{1 \times 10} \\  
\mathbf{0}_{1 \times 4} \;\;\;\; f_2(z) \mathbf{1}_{1 \times 8} \;\;\;\; \mathbf{0}_{1 \times 3} \\ 
f_3(z) \end{bmatrix} u,
\end{equation}
where
\begin{itemize}
\item $y_1$ is the sum of the simple moving averages over the past $60$ min for zones $1$ to $5$, i.e., 
$f_1(z) = \frac{1}{20} \sum_{k=1}^{20} z^{-k}$,
\item $y_2$ is $\sum_{i=5}^{12} f_2 u_i$, with $f_2$ a finite impulse response low-pass filter with Gaussian shaped impulse 
response of length $20$, obtained using Matlab's function \texttt{gaussdesign(0.5,2,10)}.  %pictured on Fig. ?? \todo{should we introduce FIR in intro notation?}
\item $y_3$ is the scalar output of a $1 \times 15$ MISO filter $f_3$ designed to forecast at each period $t$ the average total number 
of events per time-period that will occur in the whole building during the window $[t+60 \text{ min}, t+90 \text{ min}]$. 
This filter was constructed by identifying an ARMAX model \cite{Ljung98_sysIdBook} between the $15$ inputs (plus a scalar white noise) 
and the desired outputs, with the calibration done using one part of the dataset. The model chosen is takes the form
\[ 
y_{3,t} = \sum_{i=1}^4 a_i y_{3,t-i} + b_0 u_t + b_1 u_{t-2} + c_1 e_t,
\]
where $a_1, \ldots, a_4$ and $b_0, b_1$ are scalar and row vectors respectively forming the filter $f_3$, $c_1$ is a scalar 
and $e_t$ is a zero-mean white noise input postulated by the ARMAX model for system identification purposes.
\end{itemize}
%Hence our desired filter has the structure
%\begin{equation}	\label{eq: MERL desired F}
%F=\begin{bmatrix}
%* & * & * & * & 0 & 0 & 0 & 0 & 0 & 0 \\
%0 & 0 & * & * & * & * & * & 0 & 0 & 0 \\
%* & * & * & * & * & * & * & * & * & *  
%\end{bmatrix},
%\end{equation}
%where $*$ denotes a non-zero transfer function.
Fig. \ref{fig: MERL forecast} shows sample paths over a $72$h period of the $2$nd and $3$rd outputs
of the desired filter and of its ($\ln 5, 0.05$)-differentially private approximation obtained using the ZFE
mechanism. The $15$ optimal pre-filters were obtained approximately via least-squares fit of 
$\sqrt{|F_i(e^{j\omega})|_2}$ with negligible approximation error (Matlab's function $\texttt{yulewalk}$ 
implementing the Yule-Walker method \cite{Stoica:Book05:SpectralAnalysis}), rather than true 
spectral factorization mentioned in Theorem \ref{thm: error for ZFE mechanism MIMO diagonal}.
One apparent aspect of the privacy-preserving outputs seen on Fig. \ref{fig: MERL forecast} is that the 
noise level is independent of the size of the desired output signal, hence low signal values tend to 
be easily buried in the noise. This is one drawback of mechanisms relying on global sensitivity measures 
and additive noise. Another noticeable element is the fact that the noise remaining on each output can have 
quite different characteristics depending on the desired filter $F$, with the post-filter $FG^{-1}$ removing 
more high-frequency components on the second output than on the third.

%Notice on the figure that the approach used here relying on the
%notion of sensitivity requires a noise level independent of the size of the desired output signal, hence
%low signal values tend to be easily buried in the noise.
%The first example uses the dataset described in \cite{MERLdataset_2007}.
%The rationale to equip buildings with sensor networks is typically to forecast occupancy in order to optimize energy 
%efficiency. At the same time, such networks can be used to track individuals and groups of people throughout the building
%\textcolor{red}{[ref?]}, which could make the workplace less desirable to employees that know their activities are being
%constantly monitored.
%
%\begin{itemize}
%\item Moving average.
%\item ARMA forecasting.
%\end{itemize}

%\begin{figure}
%\includegraphics[width=\linewidth]{graphics/forecast.pdf}
%\caption{Real-time forecast of the total number of events that will be detected in the next 20 minutes in the whole building.
%We show the output of an ARMAX model calibrated on the dataset, which is the $3$rd output of the filter
%$F$ in \eqref{eq: MERL desired F}. We also show a $(1,0.05)$-differentially private output.}
%\label{fig: MERL forecast}
%\end{figure}

\begin{figure}
\centering
\includegraphics[width=0.7\linewidth]{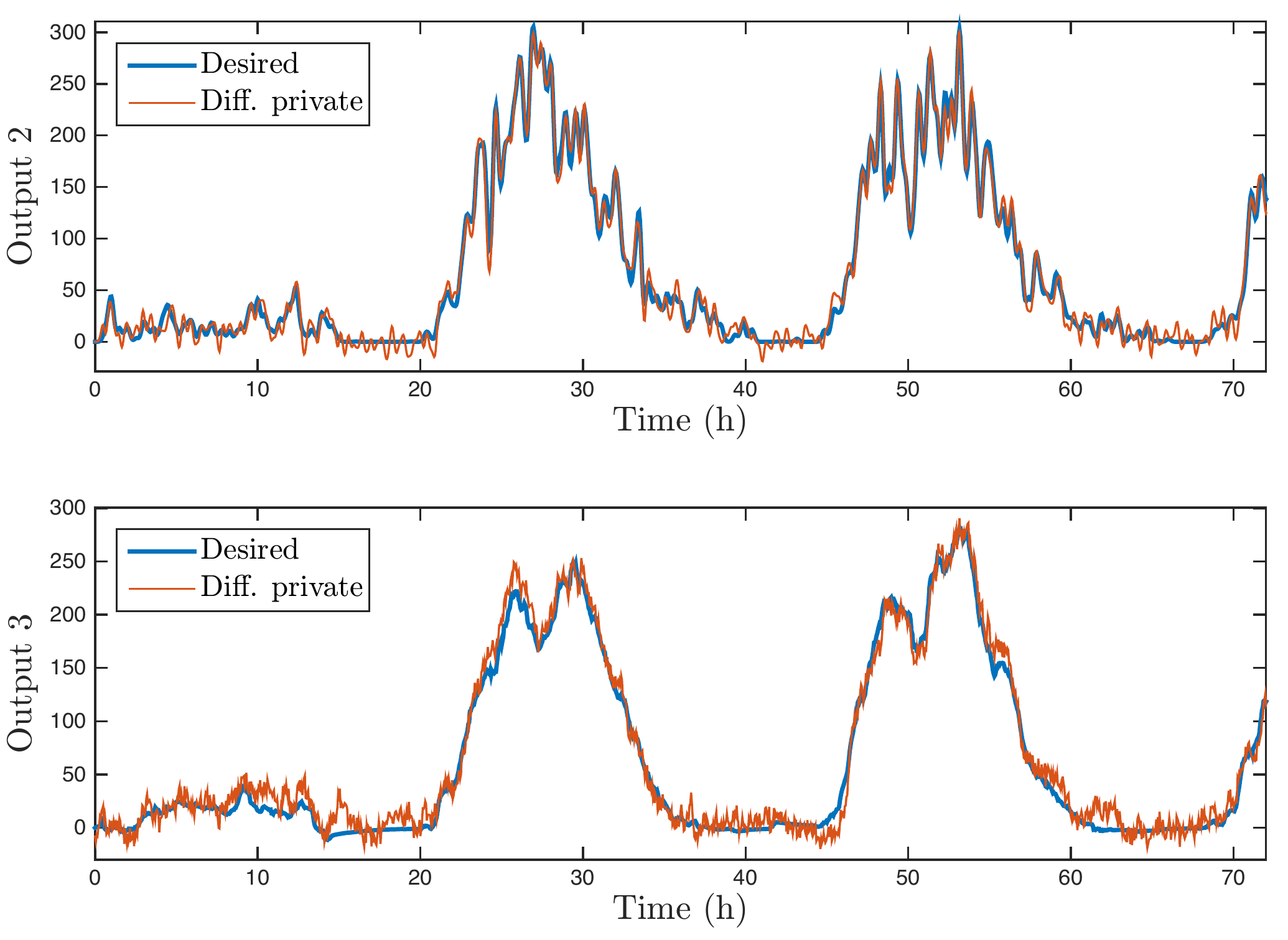}
\caption{Sample-paths over $72$ hours (the sampling period is $3$ min) for the outputs $2$ and $3$ of our differentially private approximation 
of filter (\ref{eq: MERL desired F}), shown together with the desired outputs. The privacy-related parameters are $\epsilon = \ln 5, \delta = 0.05, k_i = 4$ for $1 \leq i \leq 15$.}
%Real-time forecast of the total number of events that will be detected in the next 20 minutes in the whole building.
%We show the output of an ARMAX model calibrated on the dataset, which is the $3$rd output of the filter
%$F$ in \eqref{eq: MERL desired F}. We also show a $(1,0.05)$-differentially private output.}
\label{fig: MERL forecast}
\end{figure}

%!TEX root =  ../dpEventFiltering.tex

%\section{Linear Equalization Mechanisms}		\label{eq: linear equalization mechanisms}

%\section{Linear MMSE Mechanism}		\label{section: MMSE mechanism}
%\section{Mechanisms Exploiting Statistics on the Input Signals}
\section{Exploiting Additional Information on the Input Signals}		\label{section: additional info}

The main issue with SISO zero-forcing equalizers is the noise amplification % in communication systems
at frequencies where $|G(e^{j \omega})|$ is small, due to the inversion in $H = F G^{-1}$ \cite{Proakis00_digitalComBook}. 
This issue is not as problematic for the optimal ZFE mechanism, since in this case 
%we essentially have $|H(e^{j \omega})|=\sqrt{|F(e^{j \omega})|}$, i.e., 
the amplification is compensated by the fact that $|F(e^{j \omega})|$ and $|G(e^{j \omega})|$ given in
Theorem \ref{thm: error for ZFE mechanism MIMO diagonal} are both small at the same frequencies.
Nonetheless, we expect to be able to improve on the ZFE mechanism by using more advanced equalization schemes 
in the design of $H$. 
In this section the choice of post-filter $H$ of Fig. \ref{fig: approximation setup} depends 
on certain input signal properties, and we concentrate on the optimization of \emph{diagonal} pre-filters $G$
as for the ZFE mechanism.
Note however that the mechanisms below require that the input signals satisfy certain constraints, 
such as wide-sense stationarity, and that some publicly available information on these signals be 
available, e.g., their second order statistics. 
%These constraints might not necessarily be satisfied or this information might not be available, and 
Hence, the ZFE mechanism remains generally useful due to its broad applicability.
% and publicly available information about the input signals $u$, which might not always be available. Hence the ZFE mechanism remains generally useful.

%\subsection{Exploiting Information on Second-Order Statistics with Linear MMSE Mechanisms}
\subsection{Exploiting Information on Second-Order Statistics with Linear Mean Square Mechanisms}	\label{section: LMS mechanism}
 
%\todo[inline]{Should clarify that power spectra are assumed continuous, i.e., no Dirac delta. I.e., regular process, no periodic component.
%See book on signal processing. This is also the case for Markov chains, they should be aperiodic (and probably irreducible, so ergodic).} 

%\todo[inline]{Renaming LMMSE to LMS. Check if terminology makes sense.}

First, one can improve on the ZFE mechanism if some information on the statistics of the input signal $u$ is publicly available.
In general, constructing the optimum maximum-likelihood estimate of $\{(Fu)_k\}_{k \geq 0}$ from $\{v_k\}_{k \geq 0}$ on 
Fig. \ref{fig: approximation setup} is computationally intensive and requires the knowledge of the full joint probability distribution 
of $\{u_k\}_{k \geq 0}$ \cite{Proakis00_digitalComBook}. 
%This is the main reason why simpler linear estimation architectures such as the one described in the 
%previous subsection are more often implemented. 
Hence, we explore in this subsection a simpler scheme based on linear minimum mean square error estimation, % (LMMSE)
which we call the Linear Mean Square (LMS) mechanism. 
%It performs better than the ZFE mechanism but requires some additional knowledge about the second order statistics of $u$.
%This scheme was briefly discussed in \cite{LeNy_DP_CDC12} in the SISO case, but the optimization of $G$ described below 
%was not performed in that paper.

%\todo[inline]{Review terminology, presentation of notation (move to intro?), etc. for following paragraph.}
To develop the LMS mechanism, we assume that it is publicly known that $u$ is wide-sense stationary (WSS) 
with know mean vector $\mu$ and matrix-valued autocorrelation sequence $R_u[k] = \mathbb E[u_t u^T_{t-k}] = R_u[-k]^T, \forall k$. 
Since the mean of the output $y$ is then known, equal to $F(1) \mu$, we can assume without loss of generality that $\mu = 0$.
%Without loss of generality, we can assume $\mu$ to be zero, by substracting the known mean of $y$, equal to $F(1) \mu$.
The $z$-spectrum of $u$, $P_u(z) = \sum_{k=-\infty}^\infty R_{u}[k] z^{-k}$, is assumed for simplicity of exposition to be rational 
and positive definite on the unit circle, i.e., for $z = e^{j \omega}$.
%
%The power spectral density matrix of $u$ is denoted $P_u$, and is assumed to be have rational entries and be positive definite
%on the unit circle. \todo{review this condition, see Kailath's book.} 
More generally, given two vector-valued WSS zero-mean signals $u$ and $v$, we denote the cross-correlation matrix
$R_{uv}[k] = \mathbb E[u_t v^T_{t-k}]$, the cross z-spectrum $P_{uv}(z) = \sum_{k=-\infty}^\infty R_{uv}[k] z^{-k}$,
and all $z$-spectra are assumed to be rational.

The design of the LMS mechanism relies on a Wiener filter $H$ to estimate $y$ from $v$ 
on Fig. \ref{fig: approximation setup} \cite{Kailath00_linearEstBook, Poor94_SPbook}. 
Recall that the Wiener filter produces an estimate $\hat y$ minimizing the MSE between
$y$ and $\hat y$ over linear filters, assuming that the signal $v$ is stationary. Its design requires the
knowledge of the second-order statistics of $v$, which can be expressed in terms of those of $u$, $w$, and
of the transfer function $G$. Our design procedure involves the following steps.
%
%Moreover, replacing $F$ by $F \Phi_u$, with $\Phi_u$ the spectral factor of $P_u$, we can in fact assume that $u$ is an iid
%sequence with unit variance. \todo{really?}
%$\mathcal E_u := \mathbb E[u_k^2]$. 
First, for tractability reasons, we assume initially that $H$ is an infinite impulse response (IIR) Wiener \emph{smoother}, i.e., non-causal.
The reason is that we can then express the estimation performance analytically as a function of $G$.
%derive the performance of the non-causal infinite impulse response Wiener filter, which can
%be expressed analytically, and 
%
We then optimize this performance measure over diagonal pre-filters $G$. %, assuming this choice for $H$. 
Once $G$ is fixed, real-time considerations force us to use a lower-performance design with $H$ a \emph{causal} Wiener filter, 
or perhaps a slightly non-causal filter introducing a small delay is tolerable for a specific application.
%if it can be implemented by introducing a delay that is tolerable for a specific application.
%causal or perhaps slightly non-causal by introducing a small delay.
%

\subsubsection{Diagonal Pre-Filter Optimization}		\label{section: MMSE optimization}

The (non-causal) Wiener smoother $H$ has the transfer function
$
%H(e^{j \omega}) = \frac{P_{yv}(e^{j \omega})}{P_v(e^{j \omega})},
%H(z) = \frac{P_{yv}(z)}{P_v(z)},
H(z) = P_{yv}(z) P_v(z)^{-1}
$ 
%where $P_{yv}$ is the cross power spectral density of $y$ and $v$.
%\cite{Kailath00_linearEstBook, Poor94_SPbook}. \todo{give ref to p. number}
% N.B.: y is vector valued, v is scalar. Pyv(z) is indeed a column transfer matrix.
\cite[Section 7.8]{Kailath00_linearEstBook}.
According to Theorem \ref{thm: diagonal system sensitivity}, for $G$ diagonal we can take the 
privacy-preserving noise $w$ to be white and Gaussian with covariance $\sigma^2 I_m$ with
$\sigma^2  = \kappa_{\delta,\epsilon}^2 \|GK\|_2^2$. Since $u$ and $w$ are uncorrelated, 
we have 
\begin{equation}	\label{eq: spectra}
P_{yv}(z) = F(z) P_u(z) G(z^{-1})^T, \,\, P_v(z) = G(z) P_u(z) G(z^{-1})^T + \sigma^2 I_m.
\end{equation}
Hence 
\begin{align}	\label{eq: Wiener smoother}
H(z) = F(z) P_u(z) G(z^{-1})^T \left(G(z) P_u(z) G(z^{-1})^T + \kappa_{\delta,\epsilon}^2 \|G\|_2^2 \, I_m \right)^{-1}.
\end{align}

%Since $w$ and $u$ are uncorrelated, we have
%\[
%P_{yv}(e^{j \omega}) = P_y(e^{j \omega}) G^*(e^{j \omega}).
%P_{yv}(z) = P_u(z) F(z) G^*(1/z^*).
%P_{yv}(z) = F(z) P_u(z) G(z^{-1})^T.
%\]
%As for $P_v$, we have, with $n$ a white noise of variance $\sigma^2 = d^2 \kappa_{\delta,\epsilon}^2 \|G\|_2^2$,
%that $
%P_u(e^{j \omega}) = P_u(e^{j \omega}) |G(e^{j \omega})|^2 + \sigma^2.
%P_v(z) = P_u(z) G(z) G^*(1/z^*) + \sigma^2.
%P_v(z) = P_u(z) G(z) G(z^{-1}) + \sigma^2.
%$
%Hence %for filter with real coefficients
%\begin{align}		\label{eq: LMMSE filter}
%H(z) = \frac{P_u(z) G^*(1/z^*)}{P_u(z) G(z) G^*(1/z^*) + \kappa(\delta,\epsilon)^2 \|G\|_2^2}, \\
%H(z) = \frac{P_u(z) F(z) G(1/z)}{P_u(z) G(z) G(1/z) + \kappa(\delta,\epsilon)^2 \|G\|_2^2}, \\
%H(z) = \frac{P_u(z) F(z) G(z^{-1})}{P_u(z) G(z) G(z^{-1}) + \kappa_{\delta,\epsilon}^2 \|G\|_2^2}.
%H(e^{j \omega}) = \frac{P_u(e^{j \omega}) F(e^{j \omega}) G^*(e^{j \omega})}{P_u(e^{j \omega}) |G(e^{j \omega})|^2 + \kappa(\delta,\epsilon)^2 \|G\|_2^2}.
%\end{align}
The MSE can then be expressed as 
$e^{lms}(G) = \frac{1}{2 \pi} \int_{-\pi}^\pi \Tr (P_y(e^{j \omega}) - P_{\hat y}(e^{j \omega})) d \omega$
\cite[Chapter 7]{Kailath00_linearEstBook}. In our case,
\begin{align*}
P_{\hat y}(e^{j \omega}) &= H(e^{j \omega}) P_v(e^{j \omega}) H(e^{j \omega})^*
= P_{yv}(e^{j \omega}) P_v(e^{j \omega})^{-1} P_{yv}(e^{j \omega})^* \\
P_{\hat y} &= F P_u G^* (\sigma^2 I_m + G P_u G^*)^{-1} G P_u F^*,
\end{align*}
where on the second line and below we omit the argument $e^{j \omega}$ next to all matrices, to simplify the notation.
We have then
\begin{align*}
P_y - P_{\hat y} &= F P_u F^* - P_{\hat y} = F (P_u - P_u G^* (\sigma^2 I_m + G P_u G^*)^{-1} G P_u) F^* \\
&= F \left( P_u^{-1} + \frac{1}{\sigma^2} G^* G \right)^{-1} F^*,
\end{align*}
with the last expression obtained using the matrix inversion lemma. Finally, defining
$\tilde G(e^{j \omega}) := \frac{1}{\|GK\|_2} G(e^{j \omega}) K$, we obtain the expression
\begin{align}
e^{lms}(\tilde G) &= \frac{1}{2 \pi} \int_{-\pi}^\pi \Tr \left[ F(e^{j \omega}) 
\left( P_u(e^{j \omega})^{-1} + \frac{1}{\kappa_{\delta,\epsilon}^2} K^{-1} \tilde G(e^{j \omega})^* \tilde G(e^{j \omega}) K^{-1} \right)^{-1} 
F^*(e^{j \omega})) \right] d \omega \nonumber \\
e^{lms}(\tilde G) &=  \frac{1}{2 \pi} \int_{-\pi}^\pi \Tr \left[ \tilde F(e^{j \omega}) 
(\tilde P_u(e^{j \omega})^{-1} + \tilde G(e^{j \omega})^* \tilde G(e^{j \omega}))^{-1} \tilde F(e^{j \omega})^* \right] d \omega,
\label{eq: objective lms}
\end{align}
where $\tilde F(e^{j \omega}) = \kappa_{\delta,\epsilon} F(e^{j \omega}) K$ and 
$\tilde P_u(e^{j \omega}) = \frac{1}{\kappa_{\delta,\epsilon}^2} K^{-1} P_u(e^{j \omega}) K^{-1}$.
The objective \eqref{eq: objective lms} should be minimized over all transfer functions $\tilde G$, which by definition must
satisfy the constraint 
\begin{equation}	\label{eq: constraint lms}
\|\tilde G\|^2_2 = \frac{1}{2 \pi} \int_{-\pi}^\pi \Tr (\tilde G(e^{j \omega})^* \tilde G(e^{j \omega})) d \omega = 1.
\end{equation}
Note that in \eqref{eq: objective lms} we recover the expression \eqref{eq: performance ZFE} of the performance of 
the ZFE mechanism in the limit  $\tilde P_u(e^{j \omega}) \to \infty$.
%$P_u(e^{j \omega}) >> d^2 \kappa(\delta,\epsilon)^2$.

%\begin{align}
%%\mathcal E &= \frac{1}{2 \pi} \int_{-\pi}^\pi \frac{P_u(e^{j \omega}) |F(e^{j \omega})|^2 \kappa(\delta,\epsilon)^2 \|G\|_2^2}
%%{P_u(e^{j \omega}) |G(e^{j \omega})|^2 + \kappa(\delta,\epsilon)^2 \|G\|_2^2} d \omega \nonumber \\
%%\xi^{LMMSE}(G) 
%e^{lms}(G) &= \frac{1}{2 \pi} \int_{-\pi}^\pi \frac{P_u(e^{j \omega}) |F(e^{j \omega})|^2} 
%{ \frac{P_u(e^{j \omega})}{d^2 \kappa_{\delta,\epsilon}^2} \frac{|G(e^{j \omega})|^2}{\|G\|_2^2} + 1 } d \omega. \label{eq: error IIR Wiener noncausal}
%\end{align}
%%
%Note that in (\ref{eq: objective lms}) we recover the expression \eqref{eq: performance ZFE} of the performance of the LZF mechanism in the limit  $P_u(e^{j \omega}) >> d^2 \kappa(\delta,\epsilon)^2$.

%\textcolor{red}{[Check Luenberger as ref. on waterfilling.]}
It remains to minimize the performance measure \eqref{eq: objective lms} over the choise of pre-filters $G$ satisfying \eqref{eq: constraint lms}.
First, in the case where $\tilde P_u(e^{j \omega})$ or equivalently $P_u(e^{j \omega})$ is diagonal for all $\omega$, i.e., the different input
signals are uncorrelated, we have in fact a classical allocation problem \cite{Luenberger:book69:optimization} whose solution is of 
the ``waterfilling type''. Namely, denote $\tilde P_u(e^{j \omega}) = \text{diag}(p_{1}(e^{j \omega}),\ldots, p_{m}(e^{j \omega}))$
and $X(e^{j \omega}) = \tilde G(e^{j \omega})^* \tilde G(e^{j \omega}) = \text{diag}(x_{1}(e^{j \omega}), \ldots, x_{m}(e^{j \omega}))$,
with $x_{i}(e^{j \omega}) = |\tilde g_{ii}(e^{j \omega})|^2$. 
%$\tilde G^* \tilde G = X = \text{diag}(x_{11}, \ldots, x_{mm})$,
Omitting the expression $e^{j \omega}$ in the integrals for clarity,   
\eqref{eq: objective lms} and \eqref{eq: constraint lms} read
\begin{align*}
\min_{x} \frac{1}{2 \pi} \int_{-\pi}^\pi \sum_{i=1}^m \frac{1}{\frac{1}{p_{i}}+x_{i}} |\tilde F_i|_2^2 \, d \omega 
\text{ s.t. } \frac{1}{2 \pi} \int_{-\pi}^\pi \sum_{i=1}^m x_{i} \, d \omega = 1, \;\; x_{i}(e^{j \omega}) \geq 0, \forall \omega, i,
\end{align*}
and the solution to this convex problem is
\[
x_{i}(e^{j \omega}) = \max \left \{ 0, \sqrt{ \frac{|\tilde F_i(e^{j \omega})|_2^2 }{\lambda}} - \frac{1}{p_{i}(e^{j \omega})} \right \},
\]
where $\lambda>0$ is adjusted so that the solution satisfies the equality constraint (\ref{eq: constraint lms}).
Problems of this type are discussed in the communication literature on joint transmitter-receiver optimization 
\cite{Salz:85:MIMOopt, Yang:94:TxRxOpt}, which is not too surprising in view of our approximation setup on 
Fig. \ref{fig: approximation setup}. 

When $\tilde P_u(e^{j \omega})$ is not diagonal, it is shown in \cite{Yang:94:TxRxOpt} that
the problem can be reduced to the diagonal case if $\tilde G$ can be arbitrary, however the argument does not 
%appear to 
carry through under our constraint that $\tilde G$ must also be diagonal.
%
%\textcolor{red}{[Add waterfilling solution for $P_u$ diagonal. SDP otherwise, I don't think we can use the technique of Jiang because
%$G$ needs to stay diagonal.]}
Nonetheless, one can obtain a solution arbitrarily close to the optimal one using semidefinite programming. 
% We propose the following approximate optimization procedure, and leave the design of potentially more efficient algorithms to
%future work. 
%Since $\tilde G(z)$ is diagonal, we define
%$
%X(e^{j \omega}) = \tilde G(e^{j \omega})^* \tilde G(e^{j \omega}) = \text{diag}(x_{1}(e^{j \omega}), \ldots, x_{m}(e^{j \omega})),
%$
%with $x_{i}(e^{j \omega}) = |\tilde g_{ii}(e^{j \omega})|^2$. 
First, we discretize the optimization problem at the set of frequencies 
$\omega_q = \frac{q \pi}{N}, q=0 \ldots N$. Note that all functions are even functions of $\omega$, hence we can restrict 
out attention to the interval $[0,\pi]$. Then, we define the $m(N+1)$ variables $x_{iq} = x_{i}(e^{j \omega_q})$, with $x_{iq} \geq 0$, 
and $X_q = \text{diag}(x_{1q},\ldots,x_{mq})$. %\todo{might be able to use $n$ as subscript instead of $q$, not used?}
Using the trapezoidal rule to approximate the integrals, we obtain the following optimization problem %with $\mathbf x = \{x_{iq}\}_{i,q}$
% note that we consider the integrals over [0,pi] due to the eveness of the functions
\begin{align}
\min_{ \{X_q, M_q\}_{0 \leq q \leq N}} \quad & \frac{1}{2N} \sum_{q=0}^{N-1} 
%\Tr \left[ \tilde F_{q} (\tilde P_q^{-1} + X_q)^{-1} \tilde F_{q}^* + \tilde F_{q+1} (\tilde P_{q+1}^{-1} + X_{q+1})^{-1} \tilde F_{q+1}^* \right] \\
\Tr \left[ M_q + M_{q+1} \right] \label{eq: lms optimization objective} \\
\text{s.t. } & \begin{bmatrix} M_q & \tilde F_q \\ \tilde F_q^* & \tilde P_q^{-1} + X_q \end{bmatrix} \succeq 0, \;\; 0 \leq q \leq N, \label{eq: LMI lms} \\
& \frac{1}{2N \pi} \sum_{q=0}^{N-1} \Tr [ X_q + X_{q+1} ] = 1, \text{ and } X_q \succeq 0, \;\; 0 \leq q \leq N, \nonumber
 \end{align}
where $\tilde F_q := \tilde F(e^{j \omega_q})$ and $\tilde P_q = \tilde P_u(e^{j \omega_q})$. Note that \eqref{eq: LMI lms} is equivalent to 
$M_q \succeq \tilde F_{q} (\tilde P_q^{-1} + X_q)^{-1} \tilde F_{q}^*$ by taking the Schur complement.
The optimization problem \eqref{eq: LMI lms} is a semidefinite program, and can thus be solved efficiently even for relatively fine 
discretizations of the interval $[0,\pi]$. The transfer functions $\tilde g_{ii}$ (and hence $g_{ii}$) of the filter $\tilde G$ 
can then be obtained by interpolation of the squared magnitude $x_i(e^{j \omega})$ from the $x_{iq}$ and $m$ spectral 
factorizations.

%\vspace{0.2cm}
\begin{rem}
Even if the statistical assumptions on $u$ turn out not to be correct, and even though the optimization problem is solved
approximately rather than exactly, the differential privacy guarantee of the LMS mechanism still holds and only the approximation 
quality is impacted.
\end{rem}
%\vspace{0.2cm}

\subsubsection{Causal Mechanism}		\label{section: causal MMSE}
%\textcolor{red}{HERE}
Once the diagonal pre-filter $G$ is computed by the optimization procedure described above, we construct the final mechanism by replacing 
the smoother $H$ from \eqref{eq: Wiener smoother} by a filter respecting the causality or delay constraints of the application.
%
%The previous description of the LMMSE mechanism involves a possibly non-causal filter $H$. 
%Sometimes, the anti-causal part of this filter might have a fast decreasing impulse response, in which case the scheme can be 
%implemented approximately by introducing a small delay in the release of the output signal $\hat y$. 
%Otherwise, we need to implement a causal Wiener filter $H$. 
%
%
%Assuming that $G(z)$ is chosen to be rational and with our standing assumption that $P_u$ is rational, then $P_v$ is rational as well,
%and 
Note from \eqref{eq: spectra} that $P_v(e^{j \omega}) \succ 0$ for all $-\pi \leq \omega < \pi$. 
%Hence 
Denote the canonical spectral factorization 
$
P_v(z) = L(z) P_e L(z^{-1})^T
$
where $P_e \succ 0$ and $L$ and $L^{-1}$ are analytic in the region $|z| \geq 1$ and $L(\infty) = I_m$ \cite[Section 7.8]{Kailath00_linearEstBook}.
Then the causal Wiener filter is
$
H(z) = [P_{yv}(z) L(z^{-1})^{-T}]_+ \, P_e^{-1} L(z)^{-1}, 
$
where for a linear filter $M(z)$ with (matrix-valued) impulse reponse $\{M_t\}_{- \infty \leq t \leq \infty}$, $[M(z)]_+$ denotes the causal filter with impulse response $\{M_t \mathbf 1_{\{t \geq 0\}}\}_t$.
%
%then $P_v$ is rational as well, and we can write the spectral factorization 
%$P_v(z) = Q_v(z) Q_v(z^{-1})$, with $Q_v$ minimum phase.  % \gamma_v^2 
%
%We then have
%$
%H(z) = \frac{1}{Q_v(z)} \left[  \frac{P_{yv}(z)}{Q_v(z^{-1})}   \right]_+,   % \gamma_v^2 
%$ see, e.g.,  \cite{Poor94_SPbook}.
%Here, for a linear filter $L$ with impulse reponse $\{l_t\}_{- \infty \leq t \leq \infty}$, $[L(z)]_+$ denotes the causal filter with impulse response $\{l_t \mathbf 1_{\{t \geq 0\}}\}_t$.
%Due to the more complex expression for $H$ and the resulting MSE, the design of the optimal filter $G$ in this case  is left for future work. Here, we optimize
%$G$ assuming a possibly non-causal filter $H$, and then simply modify $H$ afterwards if causality needs to be enforced.

%!TEX root =  ../dpEventFiltering.tex

\subsection{Exploiting Information on the Input Domain using Decision-Feedback Mechanisms}

%In general, constructing the optimum maximum-likelihood estimate of $\{(Fu)_k\}_{k \geq 0}$ 
% from $\{v_k\}_{k \geq 0}$ on Fig. \ref{fig: approximation setup} is 
% computationally intensive and requires the knowledge of the full joint probability 
% distribution of $\{u_k\}_{k \geq 0}$ \cite{Proakis00_digitalComBook}. This is 
% the main reason why simpler linear estimation architectures such as the one 
% described in the previous subsection are more often implemented. 
%
%So far however, we have not exploited in the estimation procedures the knowledge 
% that the input signal takes discrete values (or perhaps is even binary valued, 
% as in \cite{Dwork10_DPcounter, Chan11_DPcontinuous}).

Signals capturing event streams often take values in a discrete set, e.g., if they 
originate from various counting sensors as in Section \ref{section: building monitoring}.  
This information can be taken into account together with the previous statistical 
information by introducing a slight degree of nonlinearity in the LMS mechanism, 
using the idea of decision-feedback  equalization \cite{Proakis00_digitalComBook}. 
We call the resulting mechanism presented below a Decision-Feedback (DF) mechanism. 
Its architecture is depicted on Fig. \ref{fig: DF mechanism}.
% where a decision block and 
%a feedback loop have been added to the structure of the optimal linear smoother \eqref{eq: Wiener smoother}.
% of the Wiener smoother \eqref{eq: Wiener smoother}.
Note that compared to the optimal Wiener smoother \eqref{eq: Wiener smoother} 
in the previous section, which is of the form $F H_u$, the only structural difference is the presence 
of the decision block and the feedback loop.

\begin{figure}
\centering
\includegraphics[width=0.8\linewidth]{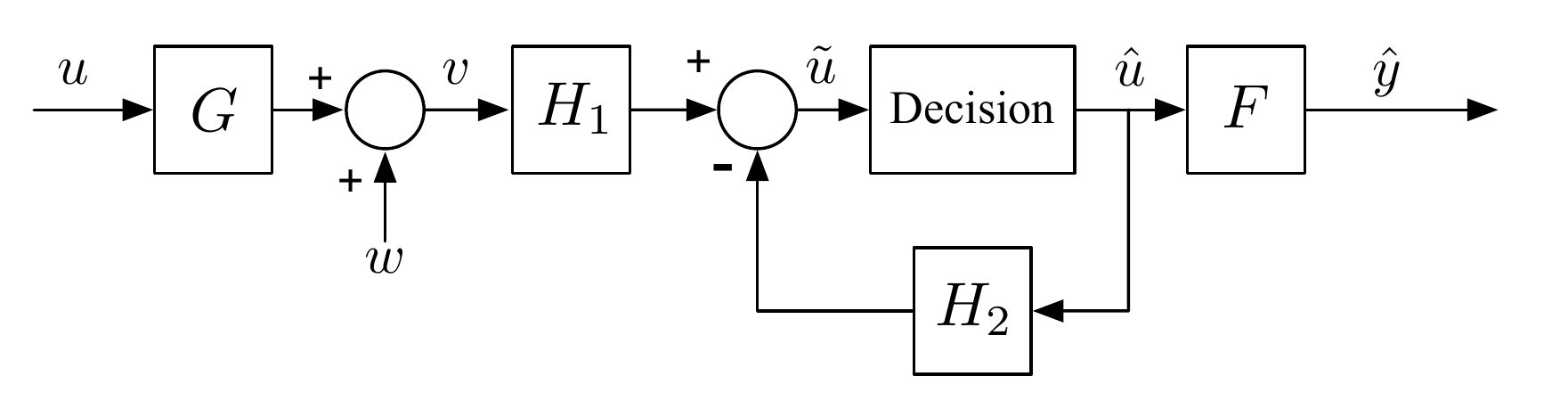}
\caption{Decision-feedback mechanism. The decision block is nonlinear and depends 
on the information about the domain of the input signal $u$, operating as a detector 
or quantizer for example.}
\label{fig: DF mechanism}
\end{figure}

The first stage of the DF mechanism is a pre-filter whose sensitivity determines 
as usual the amount of privacy-preserving noise $w$ to add. The second stage 
consists of a forward filter $H_1$, a nonlinear decision procedure (detector or 
quantizer) to estimate $u$ from $\tilde u$, which exploits the fact that $u$ 
takes discrete values, and a filter $H_2$ that feeds back the previous symbol 
decisions to correct the intermediate estimate $\tilde u$. $H_2$ is assumed to 
be strictly causal, but $H_1$ is often taken to be at least slightly non-causal in 
standard DF equalizers, for better performance \cite{Voois96_delayDFE}. In this case
the mechanism will introduce a small delay in the publication of the 
output signal $\hat y$. In the absence of detailed information about the 
distribution of $u$, the decision device can be a simple quantizer for integer 
valued input sequences, or a detector $\hat u_k = \texttt{sign}(\tilde u_k)$
 for input sequences taking values in $\{-1,+1\}$.
 % with distributional information, could perhaps design a better quantizer

The error between the desired output $Fu$ and the signal $F \tilde u$, where 
$\tilde u$ is the input of the detector, is
%\begin{align*}
$e = F(u - \tilde u) = F(u - H_1 v + H_2 \hat u)$.
%\end{align*}
For tractability reasons, the analysis and design of DF equalizers is usually 
carried out under the simplifying assumption that the past decisions entering 
$H_2$ are correct, i.e., that $\hat u_k = u_k$. In this case, the error reads
\[
e \approx F ((B - H_1 G) u - H_1 w), 
\]
with $B(z) = I+H_2(z)$ a monic filter (i.e., $B_0 = I$) since $H_2$ is strictly causal. % \todo{give def. of monic MIMO filter}
For a given $B$, the system $H_1$ minimizing the MSE is again the Wiener smoother 
to estimate $Bu$ from $Gu + w$, i.e., comparing with \eqref{eq: Wiener smoother}, 
\eqref{eq: objective lms}
\begin{align*}
H_1(z) &= B(z) P_u(z) G(z^{-1})^T \left(G(z) P_u(z) G(z^{-1})^T + \kappa_{\delta,\epsilon}^2 \|G\|_2^2 \, I_m \right)^{-1}, \\
e^{df}(B, \tilde G) &=  \frac{\kappa_{\delta,\epsilon}^2}{2 \pi} \int_{-\pi}^\pi \Tr \left[ F(e^{j \omega}) B(e^{j \omega}) 
K (\tilde P_u(e^{j \omega})^{-1} + \tilde G(e^{j \omega})^* \tilde G(e^{j \omega}))^{-1} K B(e^{j \omega})^* F(e^{j \omega})^* \right] d \omega.
\end{align*}
The next step is to optimize over the monic filter $B$. First, consider the 
spectral factorizations 
%\textcolor{red}{[ref, check that we can put the conjugate on either side...]}
% start with standard, Q R Q^*, then pose P=Q^T...
\begin{align}
& K (\tilde P_u(e^{j \omega})^{-1} + \tilde G(e^{j \omega})^* \tilde G(e^{j \omega}))^{-1} K = Q(e^{j \omega}) R Q(e^{j \omega})^* 		\label{eq: spectral factorization G} \\
& F(e^{j \omega})^* F(e^{j \omega}) = S(e^{j \omega})^* T S(e^{j \omega}),		\label{eq: spectral factorization F}
\end{align}
with $Q, S$ monic, causal, stable and invertible filters, and $R, T$ positive 
definite matrices. We have
\begin{align*}
e^{df}(B, \tilde G) 
&=  \frac{\kappa_{\delta,\epsilon}^2}{2 \pi} \int_{-\pi}^\pi \Tr \left[ 
B(e^{j \omega}) Q(e^{j \omega}) R  Q(e^{j \omega})^* B(e^{j \omega})^* S(e^{j \omega})^* T S(e^{j \omega})
\right] d\omega \\
&=  \frac{\kappa_{\delta,\epsilon}^2}{2 \pi} \int_{-\pi}^\pi \Tr \left[
T^{1/2}S(e^{j \omega})  B(e^{j \omega}) Q(e^{j \omega}) R  Q(e^{j \omega})^* B(e^{j \omega})^* S(e^{j \omega})^* T^{1/2} 
\right] d\omega \\
&= \kappa_{\delta,\epsilon}^2 \| T^{1/2} SBQ R^{1/2} \|_2^2.
\end{align*}

Note that $SBQ$ is a monic, stable and causal filter. We have the following lemma.

\begin{lem}	\label{lem: monic filter opt}
If $G(z) = \sum_{k \geq 0} G_k z^{-k}$ be a monic ($G_0 = I$), stable and causal 
filter. Let $T$, $R$ be positive definite Hermitian matrices. Then
\[
\| T^{1/2} G R^{1/2}\|^2_2 \geq \Tr(TR),
\]
with equality attained when $G = \text{Id}$, i.e., $G_k = 0$ for $k \geq 1$.
\end{lem}

\begin{proof}
We have
\begin{align*}
\| T^{1/2} G R^{1/2} \|^2_2 &= \frac{1}{2 \pi} \int_{\pi}^\pi \Tr \left[ T^{1/2} G(e^{j \omega}) R G^*(e^{j \omega}) T^{1/2} \right] d \omega \\
&= \Tr \left( \sum_{k \geq 0}  T^{1/2} G_k R G^*_k T^{1/2} \right) \;\; (\text{Parceval identity}) \\
&= \Tr (TR) +  \Tr \left( \sum_{k > 0}  T^{1/2} G_k R G^*_k T^{1/2} \right) \\
&\geq \Tr (TR),
\end{align*}
since the terms in the last sum are positive semi-definite matrices. Clearly 
equality is attained for $G_k=0$ for $k \geq 1$.
\end{proof}
From Lemma \ref{lem: monic filter opt} we deduce immediately that we should 
choose 
\[
B(z) = S^{-1}(z) Q^{-1}(z)
\] 
and the corresponding MSE is $e^{df}(\tilde G) = \kappa_{\delta,\epsilon}^2 \, \Tr(TR)$, 
with the positive definite matrices $T, R$ defined in 
\eqref{eq: spectral factorization G}, \eqref{eq: spectral factorization F}.

The final step would be to optimize the filter $G$ to minimize this expression of $e^{df}$.
Although this can be done using an approach similar to the one of the Section \ref{section: LMS mechanism},
see \cite{Yang:IT94:DF, LeNy_CDC13_eventStreamDP}, it appears that this procedure in general results 
in a pre-filter $G$ that does not depend on the query $F$, which could be an artifact of the initial assumption 
that past decisions are correct. 
This can be seen most easily in the single-input case where $T, R$ are positive scalars, 
so that $e^{df}(\tilde G) = TR$. In this product $F$ influences only in the factor $T$ 
and $G$ only in the factor $R$, hence the minimization over $G$ does not depend on $F$.
%This decoupling appears 
Optimizing $G$ independently of $F$ appears to lead to suboptimal designs in general, 
see  \cite{LeNy_CDC13_eventStreamDP}.

Hence, we propose the following design strategy for DF mechanisms. 
Note from \eqref{eq: Wiener smoother} that the (non-causal) LMS 
mechanism involves a Wiener smoother $H(z) = F(z) H_u(z)$, with $H_u$ the Linear 
Minimum Mean Squared Error estimator for $u$. We can interpret the DF mechanism 
on Fig. \ref{fig: DF mechanism} as introducing an additional (nonlinear) stage to the LMS 
mechanisms to discretize the estimate of $u$, and replacing $H_u$ by $H_1$. 
A (potentially suboptimal) strategy to improve on the performance of the LMS (or ZFE) mechanism 
is then to keep the same prefilter $G$ designed in Section \ref{section: LMS mechanism}, 
but simply replace the Wiener smoother or filter by the decision-feedback equalizer 
as described above, with a causal or almost causal approximation of the filter $H_1$. 
Our simulation results tend to confirm that good performance is achievable with this strategy.

\begin{remark}
Other DF mechanisms are possible.
For example, $H_1$ could be chosen as a zero-forcing ($H_1 = G^{-1}$ in the SISO case) rather
than a mean square equalizer, see, e.g., \cite{Belfiore:IEEE79:DFE}. 
\end{remark}

%\input{texFiles/7_examples.tex}

%\section{Related Work}

%The ZFE mechanism could be interpreted as a dynamic, causal version of the 
%matrix mechanism introduced in \cite{LiMiklau12_adaptiveMech} for static databases.

\section{Conclusion}

We have described a two-stage optimization procedure that can be used in the filtering 
of event streams in order to minimize the impact on performance of a differential privacy specification. 
The architecture considered here for the privacy-preserving mechanisms decomposes 
into a standard equalization or estimation problem, for which many alternatives techniques 
could be used depending on the scenario, and a first-stage privacy-preserving filter 
optimization problem. 
This two-stage design allows us to balance the privacy constraint and performance, 
and appears to be in fact quite general and even applicable to other definitions of privacy. 
Current work includes extending it to the design of differentially private nonlinear 
filters \cite{LeNy:TR12:privacyContraction}.

\bibliographystyle{IEEEtran/IEEEtran}


% Mac
%\bibliography{/Users/jleny/Dropbox/Research_Papers_Reports/bibtex/securityPrivacy,/Users/jleny/Dropbox/Research_Papers_Reports/bibtex/signalProcessing,/Users/jleny/Dropbox/Research_Papers_Reports/bibtex/controlSystems,/Users/jleny/Dropbox/Research_Papers_Reports/bibtex/communications,/Users/jleny/Dropbox/Research_Papers_Reports/bibtex/optimization}

% PC
\bibliography{../../../../bibtex/securityPrivacy,../../../../bibtex/signalProcessing,../../../../bibtex/controlSystems,../../../../bibtex/communications,../../../../bibtex/optimization}

\begin{thebibliography}{10}
\providecommand{\url}[1]{#1}
\csname url@samestyle\endcsname
\providecommand{\newblock}{\relax}
\providecommand{\bibinfo}[2]{#2}
\providecommand{\BIBentrySTDinterwordspacing}{\spaceskip=0pt\relax}
\providecommand{\BIBentryALTinterwordstretchfactor}{4}
\providecommand{\BIBentryALTinterwordspacing}{\spaceskip=\fontdimen2\font plus
\BIBentryALTinterwordstretchfactor\fontdimen3\font minus
  \fontdimen4\font\relax}
\providecommand{\BIBforeignlanguage}[2]{{%
\expandafter\ifx\csname l@#1\endcsname\relax
\typeout{** WARNING: IEEEtran.bst: No hyphenation pattern has been}%
\typeout{** loaded for the language `#1'. Using the pattern for}%
\typeout{** the default language instead.}%
\else
\language=\csname l@#1\endcsname
\fi
#2}}
\providecommand{\BIBdecl}{\relax}
\BIBdecl

\bibitem{LeNy_CDC13_eventStreamDP}
J.~{Le Ny}, ``On differentially private filtering for event streams,'' in
  \emph{Proceedings of the 52nd Conference on Decision and Control}, Florence,
  Italy, December 2013.

\bibitem{LeNy_CDC14_MIMOeventFiltering}
J.~{Le Ny} and M.~Mohammady, ``Differentially private {MIMO} filtering for
  event streams and spatio-temporal monitoring,'' in \emph{Proceedings of the
  53rd Conference on Decision and Control}, Los Angeles, CA, December 2014.

\bibitem{Weber10_law}
R.~H. Weber, ``Internet of things - new security and privacy challenges,''
  \emph{Computer Law and Security Review}, vol.~26, pp. 23--30, 2010.

\bibitem{PCAST14_privacy}
{President's Council of Advisors on Science and Technology}, ``Big data and
  privacy: A technological perspective,'' Report to the President, Executive
  Office of the President of the United States, Tech. Rep., May 2014.

\bibitem{EPIC03_privacyCenter}
Electronic privacy information center. Online: {http://epic.org/}.

\bibitem{Chan03_securityPrivacy_SN}
H.~Chan and A.~Perrig, ``Security and privacy in sensor networks,''
  \emph{Computer}, vol.~36, no.~10, pp. 103--105, {Oct} 2003.

\bibitem{Duncan86_disclosure}
G.~Duncan and D.~Lambert, ``Disclosure-limited data dissemination,''
  \emph{Journal of the American Statistical Association}, vol.~81, no. 393, pp.
  10--28, March 1986.

\bibitem{Sweeney02_kAnon}
L.~Sweeney, ``k-anonymity: A model for protecting privacy,''
  \emph{International Journal of Uncertainty, Fuzziness and Knowledge-Based
  Systems}, vol.~10, no.~05, pp. 557--570, 2002.

\bibitem{Sankar11_privacyInfoTheoretic}
L.~Sankar, S.~R. Rajagopalan, and H.~V. Poor, ``A theory of privacy and utility
  in databases,'' Princeton University, Tech. Rep., February 2011.

\bibitem{Xue13_securityAVN}
M.~Xue, W.~Wang, and S.~Roy, ``Security concepts for the dynamics of autonomous
  vehicle networks,'' \emph{Automatica}, vol.~50, pp. 852--857, 2014.

\bibitem{Manitara_ECC13_privacyConsensus}
N.~E. Manitara and C.~N. Hadjicostis, ``Privacy-preserving asymptotic average
  consensus,'' in \emph{Proceedings of the European Control Conference}, 2013.

\bibitem{Dwork06_DPcalibration}
C.~Dwork, F.~{McSherry}, K.~Nissim, and A.~Smith, ``Calibrating noise to
  sensitivity in private data analysis,'' in \emph{Proceedings of the Third
  Theory of Cryptography Conference}, 2006, pp. 265--284.

\bibitem{Dwork_ICAL06_DP}
C.~Dwork, ``Differential privacy,'' in \emph{Proceedings of the 33rd
  International Colloquium on Automata, Languages and Programming (ICALP)},
  ser. Lecture Notes in Computer Science, vol. 4052.\hskip 1em plus 0.5em minus
  0.4em\relax Springer-Verlag, 2006.

\bibitem{LeNy_DP_CDC12}
J.~{Le Ny} and G.~J. Pappas, ``Differentially private filtering,'' in
  \emph{Proceedings of the Conference on Decision and Control}, Maui, HI,
  December 2012.

\bibitem{LeNyDP2012_journalVersion}
------, ``Differentially private filtering,'' \emph{IEEE Transactions on
  Automatic Control}, vol.~59, no.~2, pp. 341--354, February 2014.

\bibitem{Narayanan08_netflixBreach}
A.~Narayanan and V.~Shmatikov, ``Robust de-anonymization of large sparse
  datasets (how to break anonymity of the {Netflix Prize} dataset),'' in
  \emph{Proceedings of the IEEE Symposium on Security and Privacy}, 2008.

\bibitem{Calandrino11_privacyAttackCollabFilt}
J.~A. Calandrino, A.~Kilzer, A.~Narayanan, E.~W. Felten, and V.~Shmatikov,
  ````you might also like'': Privacy risks of collaborative filtering,'' in
  \emph{Proceedings of the IEEE Symposium on Security and Privacy}, Berkeley,
  CA, May 2011.

\bibitem{Rastogi10_DPtimeSeries}
V.~Rastogi and S.~Nath, ``Differentially private aggregation of distributed
  time-series with transformation and encryption,'' in \emph{Proceedings of the
  ACM Conference on Management of Data (SIGMOD)}, Indianapolis, IN, June 2010.

\bibitem{Li11_DPcompressive}
Y.~D. Li, Z.~Zhang, M.~Winslett, and Y.~Yang, ``Compressive mechanism:
  Utilizing sparse representation in differential privacy,'' in
  \emph{Proceedings of the 10th annual ACM workshop on Privacy in the
  electronic society}, October 2011.

\bibitem{LiMiklau12_adaptiveMech}
C.~Li and G.~Miklau, ``An adaptive mechanism for accurate query answering under
  differential privacy,'' in \emph{Proceedings of the Conference on Very Large
  Databases (VLDB)}, Istanbul, Turkey, 2012.

\bibitem{Dwork10_DPcounter}
C.~Dwork, M.~Naor, T.~Pitassi, and G.~N. Rothblum, ``Differential privacy under
  continual observations,'' in \emph{Proceedings of the ACM Symposium on the
  Theory of Computing (STOC)}, Cambridge, MA, June 2010.

\bibitem{Chan11_DPcontinuous}
T.-H.~H. Chan, E.~Shi, and D.~Song, ``Private and continual release of
  statistics,'' \emph{ACM Transactions on Information and System Security},
  vol.~14, no.~3, pp. 26:1--26:24, November 2011.

\bibitem{Bolot11_DPdecayingSums}
J.~Bolot, N.~Fawaz, S.~Muthukrishnan, A.~Nikolov, and N.~Taft, ``Private
  decayed sum estimation under continual observation,'' September 2011,
  {http://arxiv.org/abs/1108.6123}.

\bibitem{LeNy14_traffic}
J.~{Le Ny}, A.~Touati, and G.~J. Pappas, ``Real-time privacy-preserving
  model-based estimation of traffic flows,'' in \emph{Proceedings of the Fifth
  International Conference on Cyber-Physical Systems (ICCPS)}, April 2014.

\bibitem{Cao13_slidingWindow}
J.~Cao, Q.~Xiao, G.~Ghinita, N.~Li, E.~Bertino, and K.-L. Tan, ``Efficient and
  accurate strategies for differentially-private sliding window queries,'' in
  \emph{Proceedings of the International Conference on Extending Database
  Technology}, 2013.

\bibitem{Fan13_mimoSeries}
L.~Fan, L.~Xiong, and V.~Sunderam, ``Differentially private multi-dimensional
  time series release for traffic monitoring,'' in \emph{27th Conference on
  Data and Applications Security and Privacy}, ser. Lecture Notes in Computer
  Science, vol. 7964.\hskip 1em plus 0.5em minus 0.4em\relax Springer, 2013,
  pp. pp 33--48.

\bibitem{Salz:85:MIMOopt}
J.~Salz, ``Digital transmission over cross-coupled linear channels,''
  \emph{{AT\&T} Technical Journal}, vol.~64, no.~6, pp. 1147--1159, July-August
  1985.

\bibitem{Yang:94:TxRxOpt}
J.~Yang and S.~Roy, ``On joint transmitter and receiver optimization for
  multiple-input-multiple-output ({MIMO}) transmission systems,'' \emph{IEEE
  Journal on Communications}, vol.~42, no.~12, pp. 3221--3231, December 1994.

\bibitem{Wilson05_tracking}
D.~H. Wilson and C.~Atkeson, ``Simultaneous tracking and activity recognition
  ({STAR}) using many anonymous, binary sensors,'' in \emph{Pervasive
  Computing}, ser. Lecture Notes in Computer Science, H.-W. Gellersen, R.~Want,
  and A.~Schmidt, Eds.\hskip 1em plus 0.5em minus 0.4em\relax Springer Berlin
  Heidelberg, 2005, vol. 3468, pp. 62--79.

\bibitem{Dwork06_DPgaussian}
C.~Dwork, K.~Kenthapadi, F.~McSherry, I.~Mironov, and M.~Naor, ``Our data,
  ourselves: Privacy via distributed noise generation,'' \emph{Advances in
  Cryptology-EUROCRYPT 2006}, pp. 486--503, 2006.

\bibitem{Wasserman10_DPstatistics}
L.~Wasserman and S.~Zhou, ``A statistical framework for differential privacy,''
  \emph{Journal of the American Statistical Association}, vol. 105, no. 489,
  pp. 375--389, March 2010.

\bibitem{Kairouz:13:DPcomposition}
P.~Kairouz, S.~Oh, and P.~Viswanath, ``The composition theorem for differential
  privacy,'' July 2015, {http://arxiv.org/abs/1311.0776}.

\bibitem{Poor94_SPbook}
H.~V. Poor, \emph{An Introduction to Signal Detection and Estimation},
  2nd~ed.\hskip 1em plus 0.5em minus 0.4em\relax Springer, 1994.

\bibitem{Ding:ML06:R1PCA}
C.~Ding, D.~Zhou, X.~He, and H.~Zha, ``{R1-PCA}: rotational invariant
  {$L_1$}-norm principal component analysis for robust subspace
  factorization,'' in \emph{Proceeding Proceedings of the 23rd international
  conference on Machine learning (ICML '06)}, 2006, pp. 281--288.

\bibitem{MERLdataset_2007}
C.~Wren, Y.~Ivanov, D.~Leigh, and J.~Westhues, ``The {MERL} motion detector
  dataset: 2007 workshop on massive datasets,'' {M}itsubishi {E}lectric
  {R}esearch {L}aboratories, Tech. Rep. TR2007-069, November 2007.

\bibitem{MERL:07:sensorNetVisu}
Y.~A. Ivanov, C.~R. Wren, A.~Sorokin, and I.~Kaur, ``Visualizing the history of
  living spaces,'' \emph{Transactions on Visualization and Computer Graphics},
  vol.~13, no.~6, pp. 1153--1160, November-December 2007.

\bibitem{Ljung98_sysIdBook}
L.~Ljung, \emph{System Identification: Theory for the User}, ser. Information
  and System Sciences.\hskip 1em plus 0.5em minus 0.4em\relax Prentice Hall,
  1998.

\bibitem{Stoica:Book05:SpectralAnalysis}
P.~Stoica and R.~L. Moses, \emph{Spectral Analysis of Signals}.\hskip 1em plus
  0.5em minus 0.4em\relax Prentice Hall, 2005.

\bibitem{Proakis00_digitalComBook}
J.~Proakis, \emph{Digital Communications}, 4th~ed.\hskip 1em plus 0.5em minus
  0.4em\relax McGraw-Hill, 2000.

\bibitem{Kailath00_linearEstBook}
T.~Kailath, A.~H. Sayed, and B.~Hassibi, \emph{Linear Estimation},
  Prentice-Hall, Ed.\hskip 1em plus 0.5em minus 0.4em\relax Prentice Hall,
  2000.

\bibitem{Luenberger:book69:optimization}
D.~G. Luenberger, \emph{Optimization by Vector Space Methods}.\hskip 1em plus
  0.5em minus 0.4em\relax New York: Wiley, 1969.

\bibitem{Voois96_delayDFE}
P.~A. Voois, I.~Lee, and J.~M. Cioffi, ``The effect of decision delay in
  finite-length decision feedback equalization,'' \emph{IEEE Transactions on
  Information Theory}, vol.~42, no.~2, pp. 618--621, March 1996.

\bibitem{Yang:IT94:DF}
J.~Yang and S.~Roy, ``Joint transmitter-receiver optimization for multi-input
  multi-output systems with decision feedback,'' \emph{IEEE Transactions on
  Information Theory}, vol.~40, no.~5, pp. 1334--1346, September 1994.

\bibitem{Belfiore:IEEE79:DFE}
C.~A. Belfiore and J.~H. Park, ``Decision feedback equalization,''
  \emph{Proceedings of the {IEEE}}, vol.~67, no.~8, pp. 1143--1156, August
  1979.

\bibitem{LeNy:TR12:privacyContraction}
\BIBentryALTinterwordspacing
J.~{Le Ny}, ``Privacy-preserving nonlinear observer design using contraction
  analysis,'' July 2015. [Online]. Available:
  \url{{http://arxiv.org/abs/1507.02250}}
\BIBentrySTDinterwordspacing

\end{thebibliography}

\end{document}